\documentclass[twocolumn]{revtex4-2}
\usepackage{amsmath,amsfonts,amsthm,amssymb,enumitem}
\theoremstyle{remark}

\begin{document}
\newtheorem*{Gt}{Generalized Gleason Theorem}
\newtheorem*{Kt}{Generalized Kraus Theorem}
\newtheorem*{Od}{Hybrid operation definition}
\newtheorem{Cor}{Corollary}
\newtheorem{Prop}{Proposition}
\newtheorem{Lem}{Lemma}
\setlist[enumerate]{wide=0pt}
\title{Probability-based approach to hybrid classical-quantum systems 
of any size: generalized Gleason and Kraus theorems}
\author{S. Camalet}
\affiliation{Sorbonne Universit\'e, CNRS, Laboratoire de Physique 
Th\'eorique de la Mati\`ere Condens\'ee, LPTMC, F-75005, Paris, France}
\begin{abstract}
Hybrid classical-quantum systems are of interest in numerous fields, from 
quantum chemistry to quantum information science. A fully quantum 
effective description of them is straightforward to formulate when the classical 
subsystem is discrete. But it is not obvious how to describe them in the general 
case. We propose a probability-based approach starting with four axioms for 
hybrid classical-quantum probability measures that readily generalize the usual 
ones for classical and quantum probability measures. They apply to discrete 
and non-discrete classical subsystems and to finite and infinite dimensional 
quantum subsystems. A generalized Gleason theorem that gives the 
mathematical form of the corresponding hybrid states is shown. This form 
simplifies when the classical subsystem probabilities are described by a 
probability density function with respect to a natural reference measure, 
for example the familiar Lebesgue measure. We formulate a requirement for 
the transformations, that is, the finite-time evolutions, of hybrid probability 
measures analogous to the complete positive assumption for quantum 
operations. For hybrid systems with reference measure, we prove a generalized 
Kraus theorem that fully determines the considered transformations provided 
they are continuous with respect to an appropriate metric. Explicit expressions 
for these transformations are derived when the classical and quantum 
subsystems are non-interacting, the classical subsystem is discrete, or the 
Hilbert space of the quantum subsystem is finite-dimensional. We also discuss 
the quantification of the correlations between the classical and quantum 
subsystems and a generalization of the quantum operations usually considered 
in the study of quantum entanglement.
\end{abstract} 
\maketitle
\section{Introduction}
Achieving a consistent description of interacting classical and quantum degrees 
of freedom is important for both theoretical and practical reasons. Such hybrid 
systems have been considered since the early days of quantum theory. In the 
Copenhagen interpretation, quantum systems interact with classical 
measurement devices \cite{Bohr,dE}. In the absence of a quantum theory of 
gravity, some progress in the understanding of the mutual interaction between 
spacetime and quantum fields may come from a hybrid approach \cite{BT,An,O}.
For fully quantum systems, the definitions of properties important in 
quantum information science involve classical notions. For instance, the 
concept of classical communication is instrumental to understand quantum 
entanglement \cite{HHHH}. From a practical perspective, a full quantum 
treatment of large molecular or condensed matter systems, for example, is 
computationally very challenging. Treating classically some degrees of freedom, 
typically the slower ones, is a common approximation \cite{KC,CB}. For 
practical uses, quantum computers are controlled by classical ones. Hybrid 
classical-quantum algorithms using both types of computers are thus needed 
\cite{CC}. 

Very different approaches have been proposed to describe hybrid 
classical-quantum systems \cite{BCGG,T1}. For a discrete classical subsystem, 
i.e., with a countable number of elementary events, a fully quantum effective 
description is widely used, in quantum information studies for example. It 
consists in treating the classical subsystem as a quantum one by representing 
each of its elementary events as a vector of a Hilbert space orthonormal basis 
\cite{D,PRA2}. Such a correspondence does not hold for continuous classical 
degrees of freedom. In this case, completely classical and completely quantum 
formal descriptions have been formulated 
\cite{E,ABCMP,A,Ge,Koo,S,BJ,PT,T2,BBGBT}. Other approaches involve usual 
classical trajectories and quantum states \cite{BT,An,CB}. Most of these 
attempts lead to significant inconsistencies such as sign changes of classical or 
quantum probabilities \cite{T1,BCGG,PT,T2}. This problem can be overcome 
but at the expense, for instance, of linearity \cite{BBGBT}. Requirements that a 
description of hybrid systems should fulfill have been proposed \cite{BT,T1}. 
Basic ones are that the classical subsystem is characterized by a positive 
probability density function, the quantum subsystem is characterized by a 
positive density operator, the equation of motion of the hybrid system is linear 
in its state, and the two subsystems evolve independently of each other 
according to usual classical and quantum dynamics when they are 
non-interacting.

In this paper, we do not follow any of the above cited approaches that 
assume a mathematical form for the states of hybrid classical-quantum 
systems. Our aim is to build a description of hybrid systems starting from 
probabilities. Such an approach guarantees, by construction, that 
probabilities are always positive. We propose basic axioms for probability 
measures of hybrid systems and for their transformations, i.e., their 
finite-time evolutions, that readily generalize those for non-hybrid systems. 
As we will see, they determine the form of the hybrid states, imply that the 
transformations are linear, and lead to the existence of particular ones that 
describe non-interacting classical and quantum subsystems.

A classical probability measure satisfies the Kolmogorov axioms, namely, it is 
a non-negative function on a set of events which is countably additive and 
equals unity for the considered sample space \cite{Ko}. Quantum theory can 
be formulated similarly. Instead of starting with quantum states and the 
Born rule for the outcome probabilities of projective measurements \cite{Bo}, 
one can first define quantum probability measures as follows. Such a measure 
is a non-negative function on the set of orthogonal projectors on a Hilbert 
space $\cal H$ which is countably additive and equals unity for the 
corresponding identity operator $I$. It follows from Gleason theorem that 
this definition leads to the usual quantum states and Born rule, provided the 
dimension of $\cal H$ is at least three \cite{G}. An analogue of this theorem 
holds for the probability measures on the set of effects on $\cal H$, which are 
the positive operators $E$ such that $I-E$ is also positive \cite{B}. 

The state transformations of quantum systems named operations are usually 
defined from unitary ones as follows. The system $q$ of interest and another 
system $q'$ are considered to be initially in a product state and to evolve 
unitarily. The resulting transformation of the state of $q$, obtained by tracing 
out $q'$, for given initial state of $q'$ and unitary evolution operator, is an 
operation. Kraus theorem shows that quantum operations can be alternatively 
characterized by the complete positivity assumption \cite{K}. This requirement 
means that, for any operation $O$ of a system $q$ and any system $q'$, there 
is a state transformation of the system made up of $q$ and $q'$ that describes 
the action of $O$ on $q$ in the presence of $q'$ but with no interaction 
between the two systems and no intrinsic evolution of the ancillary system $q'$. 
As we will see, this can be formulated at the level of probability measures.

The rest of the paper is organized as follows. In Sec.\ref{Hpm}, we present 
four axioms for hybrid classical-quantum probability measures that 
straightforwardly generalize those for classical and quantum probability 
measures. A generalized Gleason theorem that gives the mathematical form 
of the corresponding hybrid states is proved. Then, we focus on the case 
where the probability measure of the classical subsystem is characterized 
by a probability density function with respect to a natural reference measure. 
For the probability measures of such a hybrid system, a metric and a dense 
subset are introduced. In Sec.\ref{Ho}, we formulate an assumption for the 
transformations of hybrid probability measures analogous to that discussed 
above for quantum operations. After showing some properties, e.g., 
convex-linearity, of the transformations fulfilling it, that we name hybrid 
operations, we prove a generalized Kraus theorem that gives an expression for 
any such operation on the above mentioned dense subset. In Sec.\ref{So}, 
explicit expressions for hybrid operations, on the whole set of probability 
measures of the considered hybrid system, are derived when the classical and 
quantum subsystems are non-interacting, the Hilbert space of the quantum 
subsystem is finite-dimensional, or the classical subsystem is discrete. We also 
discuss how the correlations between the classical and quantum subsystems 
can be quantified and how the quantum operations usually considered in the 
context of quantum entanglement can be generalized by taking into account 
non-discrete classical information. Finally, in Sec.\ref{C}, we conclude and 
summarize our main results.
\section{Hybrid probability measures}\label{Hpm}
We consider hybrid systems consisting of classical and quantum degrees of 
freedom. The classical subsystem is characterized by a sample space $X$ and 
an event space $\cal A$ which is a $\sigma$-algebra on $X$. The quantum 
subsystem is characterized by a separable Hilbert space ${\cal H}$ and the 
set, denoted $\cal E$, of positive operators $E$ on ${\cal H}$ such that $I-E$ 
is also positive, where $I$ is the identity operator on ${\cal H}$. The elements 
of $\cal E$ are termed effects and the weak topology is considered on this set 
\cite{K}. The elements of $X$ are denoted by $x$. We stress that $X$ is an 
arbitrary set, with no particular properties.
\subsection{Axioms for hybrid probability measures and generalized 
Gleason theorem}
A probability measure $w$ of a hybrid system $h$ is a map on 
$\cal A \times \cal E$ such that, for any event $A$ and effect $E$,
\begin{enumerate}[label=(\roman*)]
\item $w(A,E)$ is real non-negative, \label{i}  
\item $w(\cup_n A_n,E)=\sum_n w(A_n,E)$ for any sequence of pairwise 
disjoint events $A_n$, \label{ii}
\item $w(A,\sum_n E_n)=\sum_n w(A,E_n)$ for any sequence of effects 
$E_n$ such that $\sum_n E_n$ is an effect, \label{iii}
\item $w(X,I)=1$. \label{iv}
\end{enumerate} 
The probability measure $p$ of the classical subsystem $c_h$ of $h$ is given 
by $p(A)=w(A,I)$. It clearly fulfills the usual Kolmogorov axioms, i.e., it is a 
non-negative $\sigma$-additive function which equals unity for $A=X$ 
\cite{Ko}. The probability measure of the quantum subsystem $q_h$ of $h$ is 
the map $E \mapsto w(X,E)$ from $\cal E$ to the unit interval. It satisfies the 
prerequisites of the Busch-Gleason theorem, i.e., it is a non-negative 
$\sigma$-additive function which equals unity for $E=I$, and so there is a 
density operator $\rho$ such that $w(X,E)=\mathrm{tr}(\rho E)$ \cite{B}. 
More generally, the map $A \mapsto w(A,E)/w(X,E)$, for a given effect $E$, is 
a classical probability measure, and the map $E \mapsto w(A,E)/w(A,I)$, for a 
given event $A$, is a quantum probability measure. These two maps can be 
interpreted as conditional probability measures given the outcome, characterized, 
respectively, by effect $E$ and event $A$, of a measurement performed, 
respectively, on $q_h$ and $c_h$. We remark that the set of orthogonal 
projectors on ${\cal H}$ can be considered instead of ${\cal E}$. In this case, 
the above expression for the probability measure of $q_h$ follows from 
Gleason theorem provided the dimension of ${\cal H}$ is at least three 
\cite{G}. For the above defined hybrid probability measures, the following 
result can be shown. 
\begin{Gt}
A map $w$ on $\cal A \times \cal E$ fulfills conditions \ref{i}-\ref{iv} if and 
only if it is given by 
\begin{equation}
w(A,E)=\int_A \mathrm{tr}(\eta(x) E) dp(x) , \label{Gt}
\end{equation}
where $p$ is a probability measure on $\cal A$ and $\eta$ is a map from $X$ 
to the set of density operators on $\cal H$ such that, for any bounded operator 
$M$ on $\cal H$, $x \mapsto \mathrm{tr}(\eta(x) M)$ is $p$-integrable. 

For given $w$, $p$ is unique and $\eta$ is $p$-almost everywhere unique.
\end{Gt}
If ensues from eq.\eqref{Gt} that the probability measure of $c_h$ is $p$ and 
that the density operator $\rho$ of $q_h$ is given by the Bochner integral 
$\int \eta dp$, where, as usual, we omit the domain of integration when it is 
equal to $X$ \cite{BI}. Any probability $w(A,E)$ can be evaluated using 
eq.\eqref{Gt}. Thus, the classical probability measure $p$ and the map $\eta$  
can be considered, together, as the state of the hybrid system $h$. Since 
$\eta(x)$ is a density operator for any $x$, equation \eqref{Gt} can be 
interpreted as follows. The classical random variable $x$ is distributed according 
to the probability measure $p$ and each of its values determines a quantum 
state. Equation \eqref{Gt} can be rewritten as $w(A,E)=\int p(A|x) p(E|x) dp(x)$ 
with $p(A|x)=1_A(x)$, where $1_A$ denotes the indicator function of the set 
$A$, and $p(E|x)=\mathrm{tr}(\eta(x) E)$, which clearly shows that the 
correlations between $c_h$ and $q_h$ are Bell local \cite{Be}. 
\begin{proof}
Let $w$ be a map on $\cal A \times \cal E$ satisfying conditions \ref{i}-\ref{iv}. 
Define the map $p$ on $\cal A$ by $p(A)=w(A,I)$. It fulfills the Kolmogorov 
axioms, see conditions \ref{i}, \ref{ii}, and \ref{iv} and is hence a classical 
probability measure on $\cal A$. For any event $A$ such that $p(A)>0$, define 
the map $v_A$ on $\cal E$ by $v_A(E)=w(A,E)/p(A)$. For any effect $E$, $I-E$ 
is an effect, and, due to condition \ref{iii}, $p(A)=w(A,E)+w(A,I-E)$. So, 
$v_A (E)$ belongs to the unit interval and $v_A (I)=1$. Moreover, it follows 
from \ref{iii} that $v_A(\sum_n E_n)=\sum_n v_A(E_n)$ for any sequence of 
effects $E_n$ such that $\sum_n E_n$ is an effect. Consequently, due to 
Busch-Gleason theorem, there is a density operator $\rho_A$ on $\cal H$ such 
that $v_A(E)=\mathrm{tr}(\rho_A E)$ \cite{B}. Let us define the map $\tau$ 
from $\cal A$ to the set ${\cal T}$ of trace-class operators on $\cal H$ by 
$\tau(A)=p(A)\rho_A$ when $p(A)>0$ and by $\tau(A)=0$ otherwise. The 
operator $\tau(A)$ is trace-class since it is positive and 
$\Vert \tau(A) \Vert=\mathrm{tr}\tau(A)=p(A)$ is finite, where 
$\Vert \cdot \Vert$ denotes the trace norm. When $p(A)=0$, $w(A,E)$ vanishes 
for any effect $E$, and so $w(A,E)=\mathrm{tr}(\tau(A) E)$ in any case. 

The set ${\cal T}$ is a Banach space with respect to the trace norm \cite{BS}. 
As it is a separable dual space, it has the Radon-Nikodym property \cite{RN}. 
The map $\tau$ satisfies the corresponding prerequisites, as we show now. 
Consider any orthonormal basis $\{ |k\rangle \}_k$ of $\cal H$ and disjoint 
events $A$ and $B$. The above derived expression for $w$ and condition 
\ref{ii} imply 
\begin{multline}
\langle k | \tau(A \cup B) | k \rangle
=w(A \cup B,| k \rangle\langle k |) \\
=w(A,| k \rangle\langle k |)+w(B,| k \rangle\langle k |)=
\langle k | \tau(A)+ \tau(B) | k \rangle .  \nonumber
\end{multline}
This equality and similar ones with $(| k \rangle+| l \rangle)/\sqrt{2}$ and 
$(| k \rangle+i| l \rangle)/\sqrt{2}$, in place of $| k \rangle$, where 
$| k \rangle$ and $| l \rangle$ are any elements of the considered basis, give 
$\langle k | \tau(A \cup B) | l \rangle=\langle k | \tau(A)+ \tau(B) | l \rangle$ 
and so $\tau(A \cup B)=\tau(A)+ \tau(B)$, i.e., $\tau$ is finitely additive. 
Consider any infinite sequence of pairwise disjoint events $A_n$ and 
$A=\cup_n A_n$. Since $\tau$ is finitely additive, $p$ is a measure, and 
$\Vert \tau(B) \Vert=p(B)$ for any event $B$, one has, for any integer $N$,
$$\Big\Vert \tau(A)-\sum_{n<N} \tau(A_n) \Big\Vert 
=\big\Vert \tau(B_N) \big\Vert=p(A)-\sum_{n<N} p(A_n) , $$
where $B_N=\cup_{n \ge N} A_n$. This distance vanishes as $N$ goes to 
infinity, and so $\tau$ is a countably additive vector measure. By definition, 
$\tau$ is absolutely continuous with respect to $p$, i.e., $\tau(A)=0$ for any 
$p$-null event $A$. The value of its variation at event $A$ is the supremum 
of $\sum_n \Vert \tau (A_n) \Vert$ over all finite partitions of $A$ into subsets 
$A_n \in \cal A$. It equals $p(A)$ since $\tau (A_n)$ is a positive operator with 
trace equal to $p(A_n)$ and $p$ is a measure. So, $\tau$ is of bounded 
variation, i.e., its variation at $X$ is finite. The above obtained properties of 
$\tau$ ensure that $\tau(A)=\int_A \eta dp$ can be written as a Bochner 
integral \cite{RN}. The map $\eta : X \rightarrow {\cal T}$ is such that there 
exists a sequence of simple functions 
$\hat \eta_n=\sum_r \eta_{n,r} 1_{A_{n,r}}$, 
where the sum is finite, $\eta_{n,r}$ is a trace-class operator, and the events 
$A_{n,r}$ are pairwise disjoints, for which 
$\lim_{n\rightarrow \infty} \Vert \eta(x)-\hat \eta_n(x) \Vert = 0$ for any $x$.
Furthermore, it fulfills $\int \Vert \eta(x) \Vert dp(x) <\infty$. Conversely, a 
map satisfying these two conditions is $p$-Bochner integrable \cite{BI,RN}. 

It can be shown that there is such a map $\eta$ for which, moreover, $\eta(x)$ 
is a density operator for any $x$, as follows. Let us first prove that $\eta(x)$ 
and $\hat \eta_n(x)$ can always be chosen to be self-adjoint. We define 
$\eta^\dag : x \mapsto \eta(x)^\dag$ and 
$\hat \eta^\dag_n=\sum_r \eta_{n,r}^\dag 1_{A_{n,r}}$. 
Since, for any $x$, $\Vert \eta^\dag(x) \Vert = \Vert \eta(x) \Vert $ and 
$\Vert \eta^\dag(x)-\hat \eta_n^\dag(x) \Vert 
= \Vert \eta(x)-\hat \eta_n(x) \Vert $, $\eta^\dag$ is a $p$-Bochner integrable 
function from $X$ to $\cal T$. Consider any elements $| k \rangle$ and 
$| l \rangle$ of an orthonormal basis of $\cal H$ and the linear form 
$T \mapsto \langle k | T | l \rangle$ on ${\cal T}$. This map is continuous 
since $|\langle k | T | l \rangle| \le \Vert T \Vert$. This inequality follows from 
$|\mathrm{tr} (MT)| \le \Vert M \Vert_{op} \Vert T \Vert$ for any trace-class 
operator $T$ and bounded operator $M$ \cite{BS}, where 
$\Vert \cdot \Vert_{op}$ denotes the operator norm, and 
$\Vert | l \rangle \langle k | \Vert_{op}=1$. Consequently, one has, for any 
event $A$, $\langle k | \tau^\dag(A) | l \rangle
=\int_A \langle k | \eta^\dag(x) | l \rangle dp(x)$, where 
$\tau^\dag(A)=\int_A \eta^\dag dp$, and a similar expression for 
$\langle k | \tau(A) | l \rangle$ \cite{BI}. It is hence straighforward to show 
that $$\langle k | \tau^\dag(A) | l \rangle^*
=\int_A \langle k | \eta^\dag(x) | l \rangle^* dp(x)
=\langle k | \tau(A) | l \rangle^*,$$
which leads to $\tau(A)=\tau^\dag(A)=\int_A \eta' dp$, where 
$\eta'=(\eta+\eta^\dag)/2$ takes self-adjoint values. Furthermore, for any $x$, 
$\eta'_n (x) =\sum_r \eta'_{n,r} 1_{A_{n,r}}(x)$, where
$\eta'_{n,r}=(\eta^{\phantom{\dag}}_{n,r}+\eta^\dag_{n,r})/2$, is 
self-adjoint and goes to $\eta'(x)$ as $n \rightarrow \infty$. 

Let $|k \rangle$ be any normalized element of $\cal H$. The above implies that, 
for any event $A$, $\langle k | \tau(A) | k \rangle
=\int_A \langle k | \eta'(x) | k \rangle dp(x)$, and that, for any $x$,
\begin{equation}
\big|\langle k | \eta'(x) | k \rangle-\langle k | \eta'_n(x) | k \rangle \big| 
\le f_n(x) , \label{in}
\end{equation}
where $f_n(x)=\Vert \eta'(x)-\eta'_n(x) \Vert$. Let us denote as 
$A_{|k\rangle}$ the set of elements $x$ of $X$ such that 
$\langle k | \eta'(x) | k \rangle<0$. The fact that $\tau(A)$ is positive for any 
event $A$ gives $\langle k | \tau(A_{|k\rangle}) | k \rangle \ge 0$ and so 
$p(A_{|k\rangle})=0$. As $\eta'_{n,r}$ is trace-class, it is compact. Since it is 
furthermore self-adjoint, there exists an orthonormal basis 
$\{ | n,r,s \rangle \}_s$ of $\cal H$ made of eigenvectors of $\eta'_{n,r}$ 
\cite{BS}. We denote the corresponding eigenvalues by $\lambda_{n,r,s}$. 
Let $B_{n,r}=\cup_s A_{|n,r,s \rangle}$ and consider any 
$x \in A_{n,r} \backslash B_{n,r}$. Such an element does not belong to any 
of the sets $A_{|n,r,s \rangle}$ and so 
$\langle n,r,s | \eta'(x) | n,r,s \rangle \ge 0$ for any $s$. Equation \eqref{in} 
and $\eta'_n(x)=\eta'_{n,r}$ then lead to 
$\lambda_{n,r,s}=\langle n,r,s | \eta'_n(x) |n,r,s \rangle \ge -f_n(x)$. This 
inequality, valid for any $s$, implies 
$\langle k | \eta'_n(x) | k \rangle \ge -f_n(x)$. Let $B=\cup_{n,r} B_{n,r}$, 
which is a $p$-null set, and consider $x \notin B$. For such an element and 
any $n$, either $\eta'_n(x)$ vanishes or there is $r$ such that 
$x \in A_{n,r} \backslash B_{n,r}$. Thus, one obtains, using again 
eq.\eqref{in}, $\langle k | \eta'(x) | k \rangle \ge -2f_n(x)$. As $| k \rangle$ is 
any normalized element of $\cal H$ and 
$\lim_{n\rightarrow\infty} f_n(x)=0$, $\eta'(x)$ is positive.

Since $|\mathrm{tr} T| \le \Vert T \Vert$ for any trace-class operator $T$, the 
linear form $T \mapsto \mathrm{tr} T$ is continuous and so 
$\mathrm{tr} \tau(A)=\int_A t dp$ for any event $A$, with 
$t : x \mapsto \mathrm{tr}\eta'(x)$ \cite{BI}. This equality, together with 
$\mathrm{tr} \tau(A)=p(A)$, yields $\int_A (1-t) dp=0$. So, the set $B'$ of 
elements $x$ of $X$ such that $t(x) \neq 1$ is $p$-null. Therefore, $\eta'(x)$ 
is a density operator for any $x \notin B''$ with $B''=B \cup B'$, which is a 
$p$-null set. The map $\eta''$ given by $\eta''(x)=\eta'(x)$ for $x \notin B''$ 
and $\eta''(x)=\eta'(x_0)$ for $x \in B''$, where $x_0$ is any element of 
$X \backslash B''$, is such that $\eta''(x)$ is a density operator for any $x$ 
and that $\int_A \eta'' dp=\tau(A)$ for any event $A$. Let $M$ be any bounded 
operator on $\cal H$. Since $|\mathrm{tr}(TM)|=|\mathrm{tr}(MT)| 
\le \Vert M \Vert_{op} \Vert T \Vert$ for any trace-class operator $T$ \cite{BS}, 
the linear form $T \mapsto \mathrm{tr}(TM)$ on $\cal T$ is continuous. So, 
$x \mapsto \mathrm{tr}(\eta''(x)M)$ is $p$-integrable and 
$w(A,E)=\int_A \mathrm{tr}(\eta''(x) E) dp(x)$ for any event $A$ and effect 
$E$ \cite{BI}.

We now show the converse. Consider a map $w$ given by eq.\eqref{Gt}. As 
$\eta(x)$ is trace-class for any $x$ and $\sqrt{E}$ is bounded for any effect 
$E$, equation \eqref{Gt} can be rewritten as 
$w(A,E)=\int_A \mathrm{tr}(\sqrt{E}\eta(x) \sqrt{E}) dp(x)$ \cite{BS}. Since, 
moreover, $\eta(x)$ is positive with unit trace and $p$ is a probability measure, 
the conditions \ref{i} and \ref{iv} are fulfilled. For any sequence of pairwise 
disjoint events $A_n$ and effect $E$, it follows from the monotone convergence 
theorem and $1_A=\sum_n 1_{A_n}$, where $A=\cup_n A_n$, that 
$\sum_n \int_{A_n} f dp= \int_A f dp$ with 
$f: x \mapsto \mathrm{tr}(\eta(x) E)$. Thus, condition \ref{ii} is satisfied. For 
any $x$, the map $\Phi : E \mapsto \mathrm{tr}(\eta(x) E)$ from $\cal E$ to 
$\mathbb{C}$ is weakly sequentially continuous. This can be seen as follows. 
Consider a sequence of effects $E_n$ converging weakly to the effect $E$, i.e., 
$\lim_{n \rightarrow \infty} |\langle \psi| \Delta_n |\phi \rangle|=0$ for any 
elements $|\psi \rangle$ and $|\phi \rangle$ of $\cal H$, where 
$\Delta_n=E_n-E$ \cite{K}. Since $\eta(x)$ is compact and self-adjoint, there 
exists an orthonormal basis $\{ | s \rangle \}_s$ of $\cal H$ made of 
eigenvectors of $\eta(x)$ \cite{BS}. We denote the corresponding eigenvalues 
by $\lambda_{s}$. They are non-negative and sum to unity as $\eta(x)$ is 
positive with unit trace. One has $|\Phi (E_n)-\Phi(E)| \le \sum_s \lambda_{s} 
|\langle s| \Delta_n |s \rangle|$. For any $\epsilon>0$, there is $S$ such that 
$\sum_{s>S} \lambda_s < \epsilon/3$ and $N$ such that 
$|\langle s| \Delta_n |s \rangle|< \epsilon/3$ for any $n >N$ and $s \le S$. As 
$E$ and $E_n$ belong to $\cal E$, $|\langle s| \Delta_n |s \rangle|\le 2$ for 
any $s$ and so $|\Phi (E_n)-\Phi(E)|<\epsilon$ for any $n>N$. The continuity 
of $\Phi$ and the monotone convergence theorem then give 
$$\sum_n \int_A \mathrm{tr}(\eta(x) E_n) dp(x)= 
\int_A  \mathrm{tr}\Big(\eta(x) \sum_n E_n \Big) dp(x) , $$ for any sequence 
of effects $E_n$ such that $\sum_n E_n$ is an effect and event $A$, which 
shows condition \ref{iii}.

For given $w$, the probability measure $p$ is unique as it is defined by 
$p(A)=w(A,I)$. Let $\eta_1$ and $\eta_2$ be any two maps satisfying 
eq.\eqref{Gt}. One has $\int_A \mathrm{tr}(\delta(x) E) dp(x)=0$ for any 
event $A$ and effect $E$, where $\delta=\eta_1-\eta_2$. Define the set 
$A_{|\psi \rangle,|\phi \rangle}
=\{x : \langle \psi| \delta(x) | \phi \rangle \ne 0 \}$ for any elements 
$|\psi \rangle$ and $|\phi \rangle$ of $\cal H$. The above equality implies that 
$p(A_{|\psi \rangle,|\psi \rangle})=0$ for any $|\psi \rangle \in \cal H$. 
Consider any elements $| k \rangle$ and $|l \rangle$ of an orthonormal basis of 
$\cal H$. The relation $$2\langle k | \delta | l \rangle
=\langle \psi| \delta |\psi \rangle-i\langle \phi|  \delta |\phi \rangle
-(1-i)(\langle k| \delta |k \rangle+\langle l| \delta |l \rangle) , $$
where the $x$-dependence of $\delta$ is omitted, 
$|\psi \rangle=| k \rangle+| l \rangle$, 
and $|\phi \rangle=| k \rangle+i| l \rangle$, 
gives $A_{|k \rangle,|l \rangle} \subset A_{|\psi \rangle,|\psi \rangle}
\cup A_{|\phi \rangle,|\phi \rangle} \cup A_{|k \rangle,|k \rangle}
\cup A_{|l \rangle,|l \rangle}$ and so $p(A_{|k \rangle,|l \rangle})=0$. 
Consequently, the set $A=\cup_{k,l} A_{|k \rangle,|l \rangle}$ is $p$-null. For 
any $x \notin A$, $\langle k| \delta(x) |l \rangle$ vanishes for any elements 
$| k \rangle$ and $|l \rangle$ of the considered basis and hence $\eta_1(x)$ 
and $\eta_2(x)$ are equal to each other.
\end{proof}
\subsection{Classical reference measure}
For many classical systems, the probability measure $p$ can be characterized 
by a probability density function with respect to a natural reference measure 
$\mu$, e.g., the counting measure if $X$ is discrete or the Lebesgue measure 
if $X$ is an Euclidean space \cite{T1}. Moreover, $\mu$ is usually 
$\sigma$-finite, i.e., $X$ is a countable union of events with finite measure, 
and so, due to Radon-Nikodym theorem, the above property of $p$ is 
equivalent to the vanishing of $p(A)$ for any event $A$ such that 
$\mu(A)=0$ \cite{Bi}.
Hybrid probability measures $w$ for which
\begin{enumerate}[label=(\roman*)]\addtocounter{enumi}{4}
\item $w(A,I)=0$ for any event $A$ such that $\mu(A)=0$, where $\mu$ is a 
$\sigma$-finite measure on ${\cal A}$, \label{v}
\end{enumerate}
are thus of particular interest. When this condition is satisfied, the hybrid 
probabilities can be written as shown below.
\begin{Cor}\label{Cor1} 
A map $w$ on $\cal A \times \cal E$ fulfills conditions \ref{i}-\ref{v} if and 
only if it is given by 
\begin{equation}
w(A,E)=\int_A \mathrm{tr}(\omega(x) E) d\mu(x) , \label{Gtmu}
\end{equation}
where $\omega$ is a map from $X$ to the set of positive trace-class operators 
on $\cal H$ such that, for any bounded operator $M$ on $\cal H$, 
$x \mapsto \mathrm{tr}(\omega(x) M)$ is $\mu$-integrable and that 
$\int \mathrm{tr}\omega(x) d\mu(x) =1$.

For given $w$, $\omega$ is $\mu$-almost everywhere unique.
\end{Cor}
Any probability $w(A,E)$ can be evaluated using eq.\eqref{Gtmu}. Thus, the 
map $\omega$ can be considered as the state of the hybrid system $h$. 
Such hybrid states are commonly used in hybrid approaches 
\cite{O,KC,T1,ABCMP,A,D}. The probability measure $p$ of $c_h$ is given by 
$p(A)=\int_A \mathrm{tr}\omega(x) d\mu(x)$ and is hence characterized by 
the probability density function $f : x \mapsto \mathrm{tr}\omega(x)$. 
The state of $q_h$ is $\rho=\int \omega d\mu$ and the quantum state for 
given $x$, such that $f(x)>0$, is $\eta(x)=\omega(x)/f(x)$.
\begin{proof}
Let $w$ be a map on $\cal A \times \cal E$ that satisfies the prerequisites of 
the corollary. The above theorem implies that $w$ is given by eq.\eqref{Gt}. 
The probability measure $p$ is given by $p(A)=w(A,I)$ and so it vanishes 
whenever $\mu$ does. Since $\mu$ is $\sigma$-finite, there is a non-negative 
measurable function $f$ such that 
$w(A,E)=\int_A \mathrm{tr}(\eta(x) E) f(x) d\mu(x)$, and, in particular, 
$\int f d\mu=1$ \cite{Bi}. Equation \eqref{Gtmu} follows with $\omega=f\eta$, 
which has all the required properties. To show the converse, consider a map 
$w$ given by eq.\eqref{Gtmu}. As $\omega(x)$ is trace-class for any $x$ and 
$\sqrt{E}$ is bounded for any effect $E$, equation \eqref{Gtmu} can be 
rewritten as 
$w(A,E)=\int_A \mathrm{tr}(\sqrt{E}\omega(x) \sqrt{E}) d\mu(x)$ \cite{BS}. 
Since, moreover, $\omega(x)$ is positive and $\mu$ is a positive measure, 
condition \ref{i} is fulfilled. Condition \ref{iv} follows from 
$\int \mathrm{tr}\omega(x) d\mu(x) =1$. For conditions \ref{ii} and \ref{iii} 
and the unicity of $\omega$, the arguments used in the proof of the above 
theorem also apply here.
\end{proof}
\subsection{Simple probability measures}
To compare hybrid probability measures fulfilling conditions \ref{i}-\ref{v}, we 
use the metric $d$ defined by
\begin{equation}
d(w_1,w_2)= \int \big\Vert \omega_1(x) - \omega_2(x) \big\Vert d\mu(x) ,
\label{d}
\end{equation}
where the maps $\omega_1$ and $\omega_2$ are given by Corollary \ref{Cor1} 
for the probability measures $w_1$ and $w_2$, respectively, and 
$\Vert \cdot \Vert$ denotes the trace norm. We now introduce a dense subset 
of the set ${\cal P}_h$ of probability measures of $h$.
\begin{Cor}\label{Corapp}
For any $\epsilon>0$ and map $w$ on $\cal A \times \cal E$ satisfying 
conditions \ref{i}-\ref{v}, there are a finite number of pairwise disjoint events 
$A_n$ and positive trace-class operators $\omega_n$ on $\cal H$ such that 
the map $\hat w$ on $\cal A \times \cal E$ given by
\begin{equation}
\hat w(A,E)=\sum_n \mathrm{tr}(\omega_n E) \mu(A \cap A_n) , \label{wapp}
\end{equation}
obeys $\hat w(X,I)=1$ and $d(w,\hat w)<\epsilon$. 
\end{Cor}
The above map $\hat w$ fulfills conditions \ref{i}-\ref{v}. We say that such a 
probability measure is simple as the corresponding map 
$\hat \omega=\sum_n \omega_n 1_{A_n}$ from $X$ to the set of positive 
trace-class operators is simple in the usual sense \cite{BI,RN}. The probability 
measure $\hat p$ of $c_h$ is given by 
$\hat p(A)=\sum_n \mathrm{tr}\omega_n \int_A 1_{A_n} d\mu$ 
and is hence characterized by the simple probability density function 
$\hat f = \sum_n \mathrm{tr}\omega_n 1_{A_n}$. The state of $q_h$ is 
$\hat \rho=\sum_n \mu(A_n) \omega_n$. For a discrete classical subsystem 
$c_h$, it can always be assumed that $X$ is a subset, finite or infinite, of 
$\mathbb{N}$, $\cal A$ is the set of all subsets of $X$, and $\mu$ is the 
counting measure. In this case, any hybrid probability measure $w$ determines 
a unique state $\omega$, given by Corollary \ref{Cor1}, and can be approached 
by a sequence of simple probability measures of the form 
$\hat w_y(A,E)=\sum_{x \le y} \mathrm{tr}(\omega(x) E/t_y) 
\mu(A \cap \{ x \})$ where $t_y=\sum_{x \le y} \mathrm{tr}\omega(x)$. 
The simple probability measures play an essential role in the following.
\begin{proof}
According to Corollary \ref{Cor1}, the map $w$ is given by eq.\eqref{Gtmu}. 
Moreover, there is a sequence of simple functions 
$\hat \omega_n=\sum_r \omega_{n,r} 1_{A_{n,r}}$, where the sum is finite, 
the operators $\omega_{n,r}$ are trace-class and self-adjoint, and the events 
$A_{n,r}$ are pairwise disjoints such that, for any $x$, 
$\lim_{n\rightarrow \infty} f_n(x)=0$, where 
$f_n(x)=\Vert \omega(x)-\hat \omega_n(x) \Vert$, see the above two proofs. 
The operator $\omega_{n,r}$ can be written as 
$\omega_{n,r}=\omega_{n,r,+}-\omega_{n,r,-}$, where $\omega_{n,r,+}$ 
and $\omega_{n,r,-}$ are positive trace-class operators such that their product 
vanishes \cite{BS}. The maps 
$\hat \omega_{n,\pm}=\sum_r \omega_{n,r,\pm} 1_{A_{n,r}}$ 
fulfill $\Vert \omega(x)-\hat \omega_{n,+}(x) \Vert 
\le f_n(x) + \mathrm{tr} \hat \omega_{n,-}(x)$ for any $x$. For any orthogonal 
projector $\Pi$ on $\cal H$, one has
$|\mathrm{tr} (\Pi \omega(x))-\mathrm{tr} (\Pi \hat \omega_n(x))| 
\le f_n(x)$ for any $x$, since $\Vert \Pi \Vert_{op} = 1$. As $\omega_{n,r}$ is 
compact and self-adjoint, there exists an orthonormal basis 
$\{ | n,r,s \rangle \}_s$ of $\cal H$ made of eigenvectors of $\omega_{n,r}$ 
\cite{BS}. Denote as $S$ the set of indices $s$ corresponding to the negative 
eigenvalues. Consider any $x$ in $A_{n,r}$. The positivity of $\omega(x)$ gives 
$\mathrm{tr} (\Pi \omega(x)) \ge 0$ with $\Pi$ the projector onto the subspace 
spanned by $\{ |n,r,s \rangle \}_{s \in S}$ and so 
$\mathrm{tr} \hat \omega_{n,-}(x) 
= - \mathrm{tr} (\Pi \hat \omega_n(x)) \le f_n(x)$. Thus, for any $x$, 
$\Vert \omega(x)-\hat \omega_{n,+}(x) \Vert$ is upperbounded by $2f_n(x)$. 

Let $\hat \omega'_n=\hat \omega_{n,+} 1_A$ where $A$ is the set of 
elements $x$ of $X$ such that $\Vert \hat \omega_{n,+}(x) \Vert 
\le 2 \mathrm{tr} \omega(x) $. These maps are simple and, for any $x$, 
$\hat \omega'_n(x)$ is trace-class and positive and 
$\Vert \omega(x)-\hat \omega'_n(x) \Vert \le \min 
\{2f_n(x) , 3 \mathrm{tr} \omega(x)\}$. Thus, this distance vanishes as 
$n \rightarrow \infty$ and, due to the dominated convergence theorem, 
$\lim_{n \rightarrow \infty} {\cal I}_n =0$, where 
${\cal I}_n= \int \Vert \omega(x)-\hat \omega'_n(x) \Vert d\mu(x)$. As a 
consequence, $t_n=\int \mathrm{tr} \hat \omega'_n(x) d\mu(x)$ converges to 
unity as $n$ goes to infinity. We define $\hat \omega''_n=\hat \omega'_n/t_n$ 
for large enough $n$. These are simple functions such that $\omega''_n(x)$ is 
positive for any $x$ and $\int \mathrm{tr} \hat \omega''_n(x) d\mu(x)=1$. 
The maps on ${\cal A} \times {\cal E}$ given by eq.\eqref{Gtmu} with 
$\hat \omega''_n$ in place of $\omega$ are of the form of eq.\eqref{wapp} 
and obey condition \ref{iv}. Moreover, since 
$\int \Vert \omega(x)-\hat \omega''_n(x) \Vert d\mu(x) 
\le {\cal I}_n + |t_n-1|$, this integral vanishes as $n\rightarrow \infty$.
\end{proof}
\section{Hybrid operations}\label{Ho}
From now on, except in subsection \ref{A}, we consider only probability 
measures satisfying the five conditions \ref{i}-\ref{v}. We are here interested 
in the maps from ${\cal P}_h$ to itself that meet the requirement stated below. 
We name them as hybrid operations. Such a transformation is said to be 
continuous if it is continuous with respect to the metric $d$ given by 
eq.\eqref{d}. 
\subsection{Definition of hybrid operations}\label{A}
Consider the bipartite system $hq$ consisting of $h$ and any quantum system, 
say $q$, with Hilbert space ${\cal H}_q$. It is a hybrid system whose events 
are those of $h$ and effects are the positive operators $G$ on 
${\cal H} \otimes {\cal H}_q$ such that $I\otimes I_q-G$ is also positive, where 
$I_q$ is the identity operator on ${\cal H}_q$. Let $w$ be any probability 
measure of $hq$. The map $w_q : F \mapsto w(X,I\otimes F)$ is the  
corresponding probability measure of $q$. For any effect $F$ of $q$ such that 
$w_q(F)$ is nonvanishing, a map $w_F$ on $\cal A \times \cal E$ can be 
defined by
\begin{equation}
w_F(A,E)=\frac{w(A,E\otimes F)}{w(X,I\otimes F)} . \label{wF}
\end{equation}
It satisfies condition \ref{iv} by construction and any of the other conditions as 
soon as $w$ does. The probabilities given by eq.\eqref{wF} can be interpreted 
as conditional probabilities given the outcome, characterized by effect $F$, of a 
measurement performed on $q$.
\begin{Od}
A map ${\cal O}_h$ from ${\cal P}_h$ to itself is a hybrid operation if and only 
if, for any quantum system $q$ with finite-dimensional Hilbert space and any 
probability measure $w$ of the bipartite system $hq$ consisting of $h$ and $q$, 
there is a probability measure $w'$ of $hq$ such that, for any event $A$ and 
effects $E$ and $F$ of $h$ and $q$, respectively,
\begin{equation}
w'(A,E\otimes F)=w_q(F){\cal O}_h(w_F)(A,E) , \label{vieq}
\end{equation}
if $w_q(F)>0$ and vanishes otherwise. \label{Od}
\end{Od}
When $h$ is quantum, i.e., when $X$ is a singleton, the above requirement is 
equivalent to the convex-linearity of ${\cal O}_h$ and the usual complete 
positivity assumption \cite{K}. It can be understood as follows. The left side of 
eq.\eqref{vieq} is the probability of obtaining the outcomes corresponding to 
$A$, $E$ and $F$ in measurements performed, respectively, on $c_h$, $q_h$, 
and $q$, for the probability measure $w'$ of $hq$. The right side of 
eq.\eqref{vieq} is the probability of obtaining the same outcomes in the same 
measurements but performed according to the following protocol. At first, the 
probability measure of $hq$ is $w$. The measurement on $q$ is made. After it, 
the probability measure of $h$ is $w_F$. It is transformed according to 
${\cal O}_h$ and finally the measurements on $c_h$ and $q_h$ are made. 
Thus, the above condition means that there exists a probability measure
transformation that describes the action of ${\cal O}_h$ on $h$ in the presence 
of $q$ with no interaction between $h$ and $q$ and no intrinsic evolution of 
$q$. The probabilities are the same whether the measurement on $q$ is made 
before or after such a transformation. For physical consistency, it must be itself 
a hybrid operation. This follows from the theorem below, provided that 
${\cal O}_h$ is continuous, see Corollary \ref{CorO}. Note that there is only 
one probability measure $w'$ fulfilling the five conditions \ref{i}-\ref{v} and 
eq.\eqref{vieq} with given $w$ and ${\cal O}_h$, see Lemma \ref{Lem3}.
\subsection{Some properties of hybrid operations}
Here, we derive properties of the operations defined above. Some of them are 
useful for the proof of the theorem below. First, observe that, since 
eq.\eqref{vieq} gives $w'_q=w_q$, the composition of two operations is also an 
operation. The following lemma states that hybrid operations are convex-linear. 
\begin{Lem}\label{convlin}
Any hybrid operation ${\cal O}_h$ satisfies
\begin{equation}
{\cal O}_h(t w_1 + (1-t) w_2 )
=t {\cal O}_h(w_1)+(1-t) {\cal O}_h(w_2) , \label{lceq}
\end{equation}
for any probability measures $w_1$ and $w_2$ and $t \in [0,1]$.
\end{Lem}
\begin{proof} 
Consider the set of $\mu$-Bochner integrable functions from $X$ to the set 
${\cal T}_h$ (${\cal T}$) of trace-class operators on ${\cal H}$ 
(${\cal H} \otimes {\cal H}_q$). The $\mu$-almost everywhere equality defines 
an equivalence relation on this set and the corresponding quotient set. We 
denote it as $\mathrm{L}_h$ ($\mathrm{L}$) and the equivalence classes as 
$[\omega]$, where $\omega \in [\omega]$. Let $\mathrm{P}_h$ 
($\mathrm{P}$) be the subset of $\mathrm{L}_h$ ($\mathrm{L}$) consisting 
of all classes $[\omega]$ containing a map $\omega$ such that $\omega(x)$ is 
positive for any $x$ and $\int \mathrm{tr}\omega(x) d\mu(x) =1$. Corollary 
\ref{Cor1}, together with Pettis measurability theorem and the fact that the set 
of bounded operators on ${\cal H}$ (${\cal H} \otimes {\cal H}_q$) is the dual 
of ${\cal T}_h$ (${\cal T}$) \cite{BS,RN}, implies that the map 
$[\omega] \mapsto w$ from $\mathrm{P}_h$ ($\mathrm{P}$) to the set of 
probability measures of $h$ ($hq$), where $w$ is given by eq.\eqref{Gtmu}, is 
bijective. Therefore, to any map ${\cal O}_h$ (${\cal O}$) from the set of 
probability measures of $h$ ($hq$) to itself there corresponds a map $O_h$ 
($O$) from $\mathrm{P}_h$ ($\mathrm{P}$) to itself and vice versa. 

Let $w$ be any probability measure of $hq$ and $[\Omega]$ the corresponding 
class in $\mathrm{P}$. The probability measure given by eq.\eqref{wF} can be 
written into the form of eq.\eqref{Gtmu} with 
$$\omega_F(x)=\mathrm{tr}_q (\Omega(x) I\otimes F)/w_q(F),$$ in place of 
$\omega(x)$, where $\mathrm{tr}_q$ denotes the partial trace with respect to 
${\cal H}_q$. As $\Vert \mathrm{tr}_q T \Vert \le \Vert T \Vert$ for any 
self-adjoint element $T$ of ${\cal T}$ and $\Omega$ is $\mu$-Bochner 
integrable, so is $\omega_F$. Moreover, $[\omega_F]$ belongs to 
$\mathrm{P}_h$ and corresponds to $w_F$. From equation \eqref{vieq} and 
the fact that $w'$ is a probability measure, it follows that
\begin{eqnarray}
w'(A,E\otimes F) 
&=&\int_A \mathrm{tr}(\Omega'(x) E\otimes F) d\mu(x) \label{eu1} \\
&=&w_q(F) \int_A \mathrm{tr}(\omega'(x) E) d\mu(x) , \label{eu2}
\end{eqnarray}
where $[\Omega']$ corresponds to $w'$ and $\omega' \in O_h([\omega_F])$. 

Let $w_1$ and $w_2$ be any probability measures of $h$ and $t \in (0,1)$. 
Consider ${\cal H}_q=\mathbb{C}^2$, $F=|\psi\rangle \langle \psi|$ where 
$|\psi\rangle=(|1\rangle+e^{i\theta}|2\rangle)/\sqrt{2}$ with 
$\{|1\rangle ,|2\rangle\}$ an orthonormal basis of ${\cal H}_q$ and 
$\theta \in [-\pi,\pi)$, and $\Omega=t \omega_1 \otimes |1\rangle \langle 1|
+(1-t) \omega_2 \otimes |2\rangle \langle 2|$ where $[\omega_1]$ and 
$[\omega_2]$ correspond to $w_1$ and $w_2$, respectively. With these 
choices, $w_q(F)=1/2$ and $\omega_F=t \omega_1 +(1-t) \omega_2$ are 
independent of $\theta$ and so is  $w'(A,E\otimes F)$, see eq.\eqref{eu2}. 
Since $[\Omega'] \in \mathrm{P}$, one has 
$\Omega'=\alpha_1 \otimes |1\rangle \langle 1|
+\alpha_2 \otimes |2\rangle \langle 2|
+\beta \otimes |1\rangle \langle 2|
+\beta^\dag \otimes |2\rangle \langle 1|$ $\mu$-almost everywhere, where 
the maps $\alpha_1$, $\alpha_2$, $\beta$, and $\beta^\dag$ from $X$ to 
${\cal T}_h$ are such that, for any $x$, $\alpha_1(x)$ and $\alpha_2(x)$ are 
positive and $\beta^\dag(x)=\beta(x)^\dag$. Differentiating both sides of 
eq.\eqref{eu1} with respect to $\theta$, one obtains 
$\int_A \mathrm{tr}((\beta(x)-e^{-2i\theta}\beta(x)^\dag)E)d\mu(x)=0$ for 
any event $A$ and effect $E$. Thus, for any $\theta$, there is a $\mu$-null set 
$A_{\theta}$ such that $\beta(x)-e^{-2i\theta}\beta(x)^\dag=0$ for any 
$x \notin A_{\theta}$, see the proof of the above theorem. This gives 
$\beta(x)=0$ for any $x \notin A_0 \cup A_{\pi/2}$ and so 
$\Omega'=\alpha_1 \otimes |1\rangle \langle 1|
+\alpha_2 \otimes |2\rangle \langle 2|$ $\mu$-almost everywhere. Let us 
define $\omega'_1=\alpha_1/t$ and $\omega'_2=\alpha_2/(1-t)$. For 
$F=|s\rangle \langle s|$, where $s=1$ or $2$, one has $w_q(F)=t+(1-2t)(s-1)$ 
and $\omega_F=\omega_s$. Hence, the equality of the right sides of 
eq.\eqref{eu1} and eq.\eqref{eu2} implies $\omega'_s \in O_h([\omega_s])$. 
This result, together with eq.\eqref{vieq} and eq.\eqref{eu1} with $F=I_q$, 
leads to ${\cal O}_h(w_{I_q})=t {\cal O}_h(w_1)+(1-t) {\cal O}_h(w_2)$ 
and a direct evaluation gives $w_{I_q}=t w_1+(1-t) w_2$.
\end{proof} 
Convex-linearity and continuity are related as follows.
\begin{Lem}\label{cont}
Any convex-linear map from the set of simple probability measures of $h$ to 
${\cal P}_h$ can be uniquely extended into a continuous map ${\cal C}_h$ on 
${\cal P}_h$. It fulfills 
\begin{equation}
d \big({\cal C}_h (w_1),{\cal C}_h (w_2) \big)\le d(w_1,w_2) , 
\label{contineq}
\end{equation}
for any probability measures $w_1$ and $w_2$. 

All convex-linear maps from ${\cal P}_h$ to itself are continuous and obey 
eq.\eqref{contineq} when ${\cal H}$ is finite-dimensional or the classical 
subsystem of $h$ is discrete.
\end{Lem}
Lemma \ref{cont} obviously applies to hybrid operations since they are 
convex-linear, as shown by Lemma \ref{convlin}. In particular, it follows from 
these two lemmas that all hybrid operations are continuous in the above 
mentioned particular cases. Moreover, the metric $d$ never increases under 
the action of continuous hybrid operations, analogously to the behavior of 
numerous metrics for quantum systems \cite{SIOR}.
\begin{proof}
Consider a convex-linear map ${\cal C}_h$ from the set of simple probability 
measures of $h$ to ${\cal P}_h$. We use the same notations as in the proof of 
Lemma \ref{convlin}. The set $\mathrm{L}_h$ is a vector space with the usual 
definitions: $[\omega_1]+[\omega_2]=[\omega_1+\omega_2]$ and 
$z[\omega_1]=[z\omega_1]$ for any $z \in \mathbb{C}$ and classes 
$[\omega_1]$ and $[\omega_2]$ of $\mathrm{L}_h$. A norm 
$\vert \cdot \vert$ can be defined on $\mathrm{L}_h$ by 
$\vert [\omega] \vert = \int \Vert \omega(x) \Vert d\mu(x)$. This simple 
notation is used since no confusion with the usual norm on $\mathbb{C}$ is 
possible. We denote as $\mathrm{H}_h$ the subset of $\mathrm{L}_h$ 
consisting of all classes containing a a map $\omega$ such that $\omega(x)$ is 
self-adjoint for any $x$ and as $\mathrm{S}_h$ the subspace of 
$\mathrm{L}_h$ consisting of all classes containing a simple function and 
define $\mathrm{P}'_h=\mathrm{P}_h \cap \mathrm{S}_h$ and 
$\mathrm{H}'_h=\mathrm{H}_h \cap \mathrm{S}_h$. 

The map $[\omega] \mapsto w$ from $\mathrm{P}'_h$ ($\mathrm{P}_h$) to 
the set of simple (all) probability measures of $h$, where $w$ is given by 
eq.\eqref{Gtmu}, is bijective. Consequently, a map $C_h$ from 
$\mathrm{P}'_h$ to $\mathrm{P}_h$ univocally corresponds to ${\cal C}_h$. 
Let us first show that $C_h$ can be extended into a bounded linear map 
$C'_h : \mathrm{S}_h \rightarrow \mathrm{L}_h$. Since the bijections 
$[\omega] \mapsto w$, their inverses, and ${\cal C}_h$ are convex-linear, so is 
$C_h$. For all classes $[\omega]$ such that $\omega(x)$ is positive for any $x$, 
we define $C'_h([0])=[0]$ and $C'_h([\omega])=\vert [ \omega] \vert 
C_h([ \omega]/\vert [ \omega] \vert)$ if $\vert [ \omega ]\vert>0$. The 
transformed class $C'_h([ \omega])$ also contains a map positive on $X$ and 
$\vert C'_h([ \omega]) \vert=\vert [ \omega] \vert$. For any such classes 
$[\omega_1]$ and $[\omega_2]$, 
since $\vert [\omega_1] \vert+\vert [\omega_2] \vert
=\vert [\omega_1+\omega_2] \vert$, the convex-linearity of $C_h$ 
implies that $C'_h([\omega_1]+[\omega_2])
=C'_h([\omega_1])+C'_h([\omega_2])$, see eq.\eqref{lceq}. Moreover, one 
has $C'_h(z[\omega_1])=zC'_h([\omega_1])$ for any positive real number $z$. 

For any class $[\hat \omega] \in \mathrm{S}_h$, one can write 
$\hat \omega=\sum_n \omega_n 1_{A_n}$, where the sum is finite, the 
operators $\omega_n$ are trace-class, and the events $A_n$ are pairwise 
disjoints and such that $\mu(A_n)$ is finite. The operator $\omega_n$ can be 
expanded in a unique way as 
$\omega_n=\alpha_{+n} - \alpha_{-n} +i \beta_{+n}- i\beta_{-n}$, 
where $\alpha_{+n}$, $\alpha_{-n}$, $\beta_{+n}$, and $\beta_{-n}$ are 
positive trace-class operators such that $\alpha_{+n} \alpha_{-n}=0$ and 
$\beta_{+n}\beta_{-n}=0$. One has 
$\Vert \alpha_{+n}-\alpha_{-n} \Vert=\Vert \alpha_{+n} \Vert 
+ \Vert \alpha_{-n} \Vert$ and a similar relation for the operators 
$\beta_{\pm n}$ \cite{BS}. We define the simple functions 
$\alpha_{\pm}=\sum_n \alpha_{\pm n} 1_{A_n}$ and 
$\beta_{\pm}=\sum_n \beta_{\pm n} 1_{A_n}$, and the image
\begin{equation}
C'_h([\hat \omega])=C'_h([\alpha_+])-C'_h([\alpha_-])
+iC'_h([\beta_+])-iC'_h([\beta_-]) . \label{Cph}
\end{equation}
For $[\hat \omega] \in \mathrm{H}'_h$, consider any positive operators 
$\gamma_n$ and $\delta_n$ such that $\omega_n=\gamma_n - \delta_n$ and 
define the corresponding simple functions $\gamma$ and $\delta$. As 
$\alpha_+ + \delta=\alpha_-+\gamma$, it follows from the above derived 
properties of $C'_h$ that 
$C'_h([\hat \omega])=C'_h([\gamma])-C'_h([\delta])$. Using this result, 
it can be shown that these properties hold on $\mathrm{H}'_h$ with any real 
$z$ and that $C'_h$ is linear on $\mathrm{S}_h$. Moreover, since 
$\vert [\alpha_+] \vert +\vert [\alpha_-] \vert
=\vert [\alpha_+ - \alpha_- ] \vert 
= \vert [(\hat \omega+\hat \omega^\dag)/2] \vert 
\le \vert [\hat \omega] \vert$, where 
$\hat \omega^\dag=\sum_n \omega^\dag_n 1_{A_n}$, and similarly for 
$[\beta_\pm]$, $C'_h$ is bounded with norm at most $2$ and 
$\vert C'_h([\hat \omega]) \vert \le \vert [\hat \omega] \vert$ 
for any $[\hat \omega] \in \mathrm{H}'_h$. The set $\mathrm{L}_h$ is a 
Banach space with the norm $\vert \cdot \vert$ and $\mathrm{S}_h$ is a 
dense subspace of it \cite{BS}. So, according to the bounded linear 
transformation theorem, $C'_h$ can be uniquely extended on 
$\mathrm{L}_h$ and its norm remains the same \cite{BLT}. For any class of 
$\mathrm{H}_h$, there is a sequence in $\mathrm{H}'_h$ converging to it, 
see the proof of Corollary \ref{Corapp}. Thus, as $C'_h$ is linear and bounded, 
the above inequality holds in $\mathrm{H}_h$. 

It remains to prove that $C'_h(\mathrm{P}_h) \subset \mathrm{P}_h$. 
Corollary \ref{Corapp} implies that 
$\mathrm{P}_h \subset \overline{\mathrm{P}_h'}$ where 
$\overline{\mathrm{P}'_h}$ is the closure of $\mathrm{P}'_h$. It can be 
shown that  $\overline{\mathrm{P}_h'} \subset \mathrm{P}_h$ as follows. 
Consider any sequence of classes $[\hat \omega_n]$ in $\mathrm{P}'_h$ that 
converges to $[\omega]$. As $\hat \omega_n(x)$ is self-adjoint for any $x$, 
one has $|[\omega-\omega^\dag]|=0$, where 
$\omega^\dag : x \rightarrow \omega(x)^\dag$, and so 
$[\omega] \in \mathrm{H}_h$. From $[\hat \omega_n] \in \mathrm{P}_h$, 
it follows that 
$|1 - \int \mathrm{tr}\omega(x) d\mu(x)| \le |[\omega-\hat \omega_n]|$. 
Hence, the left side of this inequality is zero. For any $|\psi \rangle \in {\cal H}$ 
and $\epsilon>0$, we define the sets $A_{|\psi \rangle}
=\{x : \langle \psi| \omega(x) |\psi \rangle <0\}$ 
and $A_{|\psi \rangle,\epsilon}
=\{x : \langle \psi| \omega(x) |\psi \rangle \le -\epsilon \}$. Markov's inequality 
implies that $\mu(A_{|\psi \rangle,\epsilon,n}) 
\le |[\omega-\hat \omega_n]|/\epsilon$, where $A_{|\psi \rangle,\epsilon,n}$ 
is the set of elements $x$ such that 
$|\langle \psi| \omega(x) -  \hat \omega_n(x) |\psi \rangle | \ge \epsilon$ 
\cite{Bi}. Since $\langle \psi| \hat \omega_n(x) |\psi \rangle \ge 0$ for any $x$, 
one has $A_{|\psi \rangle,\epsilon} \subset A_{|\psi \rangle,\epsilon,n}$. As 
this inculsion holds for any $n$,  $A_{|\psi \rangle,\epsilon}$ is a $\mu$-null 
set and so is $A_{|\psi \rangle}=\cup_{n=1}^\infty A_{|\psi \rangle,1/n}$. As 
$\omega$ is $\mu$-Bochner integrable, there is a sequence of simple functions 
$\gamma_n$ such that, for any $x$, $\Vert \omega(x) - \gamma_n(x) \Vert$ 
vanishes as $n$ goes to infinity. It can always be assumed that $\gamma_n(x)$ 
is self-adjoint for any $x$ since $[\omega] \in \mathrm{H}_h$. Using similar 
arguments as those following eq.\eqref{in}, one can show that $\omega(x)$ is 
positive for $\mu$-almost any $x$, and so $[\omega] \in \mathrm{P}_h$. 
Consequently, $\mathrm{P}_h$ is the closure of $\mathrm{P}_h'$. Thus, since 
$C'_h$ is continuous and $C'_h(\mathrm{P}'_h)\subset \mathrm{P}_h$, 
$C_h'(\mathrm{P}_h)$ is a subset of $\mathrm{P}_h$. The map from 
${\cal P}_h$ to itself corresponding to the restriction of $C'_h$ to 
$\mathrm{P}_h$ is continuous, obeys eq.\eqref{contineq}, and coincides with 
${\cal C}_h$ for simple probability measures. Due to Corollary \ref{Corapp}, it 
is the unique map with these properties.

Assume now that $\cal H$ is finite-dimensional. Consider a convex-linear map 
${\cal C}_h$ from ${\cal P}_h$ to itself and the corresponding map 
$C_h : \mathrm{P}_h \rightarrow \mathrm{P}_h$. It can be shown that $C_h$ 
can be extended into a bounded linear map $C'_h$ from $\mathrm{L}_h$ to 
itself as follows. For all classes $[\omega]$ such that $\omega(x)$ is positive 
for any $x$, $C'_h([\omega])$ can be defined as above. Let $[\omega]$ be 
any class in $\mathrm{H}_h$ and $(\hat \omega_n)_n$ a sequence of simple 
functions that converges pointwise to $\omega$. One can write 
$\omega(x)=\alpha_{+}(x) - \alpha_{-}(x)$ where $\alpha_{+}(x)$ and 
$\alpha_{-}(x)$ are positive trace-class operators such that their product 
vanishes \cite{BS}. We define $\alpha_\pm : x \mapsto \alpha_\pm(x)$ and 
$\tilde \omega : x \mapsto \alpha_+(x)+\alpha_-(x)$. For any $x$, 
$\tilde \omega(x)$ is the absolute value of $\omega(x)$, i.e., 
$\tilde \omega(x)$ is positive and $\tilde \omega(x)^2=\omega(x)^2$, 
and $\alpha_\pm(x)=(\tilde \omega(x)\pm\omega(x))/2$ \cite{BS}. Since 
$\Vert \tilde \omega(x) \Vert= \Vert \omega(x) \Vert$ for any $x$, 
$\int \Vert \tilde \omega(x) \Vert d\mu(x)$ is finite. 
Similar maps, denoted by $\alpha_{n\pm}$ and $\tilde \omega_n$, exist for 
$\hat \omega_n$. One has, for any $x$,
\begin{multline}
\Vert  \tilde \omega_x-\tilde \omega_{n,x} \Vert^2
\le e \mathrm{tr} \big( ( \tilde \omega_x-\tilde \omega_{n,x} )^2 \big) 
\le e \Vert  \omega_x^2-\hat \omega_{n,x}^2 \Vert \\
\le e \big(2 \Vert \omega_x \Vert 
+ \Vert  \omega_x-\hat \omega_{n,x} \Vert \big) 
\Vert  \omega_x-\hat \omega_{n,x} \Vert, \nonumber
\end{multline}
where the argument $x$ is written as an index and $e$ is the dimension of 
${\cal H}$. The first inequality results from Cauchy-Schwarz inequality, the 
second one is shown in Ref.\cite{Ineq}, and the last one is obtained using the 
properties of the trace norm. Thus, $\tilde \omega$ is $\mu$-Bochner 
integrable and so are $\alpha_+$ and $\alpha_-$. For any 
$[\omega] \in \mathrm{L}_h$, the map 
$\omega^\dag:x\mapsto \omega(x)^\dag$ is $\mu$-Bochner integrable and 
so are $(\omega+\omega^\dag)/2$ and $(\omega-\omega^\dag)/2i$. As 
these two maps take self-adjoint values, they can be written in terms of 
$\mu$-Bochner integrable maps that take positive values, denoted, respectively, 
by $\alpha_\pm$ and $\beta_\pm$ and so 
$[\omega]=[\alpha_+]-[\alpha_-]+i[\beta_+]-i[\beta_-]$. 
We define $C'_h$ on $\mathrm{L}_h$ by eq.\eqref{Cph}. It can be shown as 
above that it is linear and fulfills 
$\vert C'_h([\omega]) \vert \le \vert [\omega] \vert$ for any 
$[\omega] \in \mathrm{H}_h$. Consequently, ${\cal C}_h$ obeys 
eq.\eqref{contineq} and is hence continuous.

Assume now that the classical subsystem of $h$ is discrete, i.e., the sample 
space $X$ is a subset of $\mathbb{N}$, $\cal A$ is the set of all subsets of 
$X$, and $\mu$ is the counting measure. Consider a convex-linear map 
${\cal C}_h$ from ${\cal P}_h$ to itself and the corresponding map 
$C_h : \mathrm{P}_h \rightarrow \mathrm{P}_h$. Let $[\omega]$ be any 
class of $\mathrm{L}_h$. Maps $\alpha_\pm$ and $\beta_\pm$ can be 
defined as above. The sequence of simple functions 
$\hat \alpha_{y \pm}= \alpha_\pm 1_{\{ x : x\le y\}}$, 
where $y$ is any element of $X$, converges pointwise to $\alpha_\pm$ as 
$y \rightarrow \infty$ and thus $\alpha_\pm$ is $\mu$-Bochner integrable, 
and so is $\beta_\pm$. The bounded linear map $C'_h$ defined on 
$\mathrm{L}_h$ by eq.\eqref{Cph} coincides with $C_h$ on $\mathrm{P}_h$ 
and fulfills $\vert C'_h([\omega]) \vert \le \vert [\omega] \vert$ for any 
$[\omega] \in \mathrm{H}_h$. So, ${\cal C}_h$ obeys eq.\eqref{contineq} 
and is hence continuous.
\end{proof}
A useful consequence of the hybrid operation definition is the following.
\begin{Lem}\label{Lem3}
For any hybrid operation ${\cal O}_h$ and quantum system $q$ with 
finite-dimensional Hilbert space, there is a unique map 
${\cal O} : w \mapsto w'$ from ${\cal P}_{hq}$ to itself, where $w'$ fulfills 
eq.\eqref{vieq}. It is convex-linear for simple probability measures.
\end{Lem}
\begin{proof}
Consider the bipartite system $hq$ consisting of $h$ and $q$, any probability 
measure $w$ of $hq$, and any probability measure $w'$ that fulfills 
eq.\eqref{vieq}. One has 
$w'(A,G) =\int_A \mathrm{tr}(\Omega'(x) G) d\mu(x)$, 
for any event $A$ and effect $G$ of $hq$, where the class $[\Omega']$ of 
$\mathrm{L}$ corresponds to $w'$. Using the expansion 
$G=\sum_{k,l} \mathrm{tr}_q (I \otimes |l\rangle \langle k | G)
\otimes |k\rangle \langle l |$ where the sum runs over all couples of elements 
$|k\rangle$ and $|l\rangle$ from an orthonormal basis of the Hilbert space 
of $q$, one can show that
\begin{multline}
w'(A,G)= \sum_{u=1}^2  \sum_{k,l}
\big( v(E_{k,l,u}\otimes F_{k,l,u}) +v(E_{k,l,3} \otimes F_{k,l,3}) \\
-v(E_{k,l,u} \otimes F_{k,l,3})-v(E_{k,l,3} \otimes F_{k,l,u}) \big)  , 
\nonumber
\end{multline}
where $v(G')=w'(A,G')$ for any effect $G'$, 
$F_{k,l,1}=(| k \rangle+| l \rangle)(\langle k |+\langle l |)/2$, 
$F_{k,l,2}=(| k \rangle+i| l \rangle)(\langle k |-i\langle l |)/2$, 
$F_{k,l,3}=(|k \rangle\langle k| +|l \rangle\langle l|)/2$,
and $E_{k,l,u}=\mathrm{tr}_q(I\otimes F_{k,l,u} G)$. So, $w'$ is fully 
determined by $w$ via eq.\eqref{vieq}. We define the map 
${\cal O} : w \mapsto w'$. 

For any $[\Omega] \in \mathrm{P}'$, the class of the map 
$\omega : x \mapsto \mathrm{tr}_q (\Omega(x) I\otimes F)$, where $F$ is 
any effect of $q$, belongs to $\mathrm{S}_h$. Lemma \ref{convlin} ensures 
that ${\cal O}_h$ is convex-linear. We denote as $O'_h$ the bounded linear 
map from $\mathrm{S}_h$ to $\mathrm{L}_h$ corresponding to ${\cal O}_h$, 
see the proof of Lemma \ref{cont}. Using the linearity of $O'_h$, equation 
\eqref{eu2} can be rewritten as
$w'(A,E\otimes F)= \int_A \mathrm{tr}(\omega'(x) E) d\mu(x)$ where 
$\omega' \in O'_h([\omega])$.  Then, arguing as above, it can be shown that 
${\cal O}(w)(A,G) = \sum_{k,l} 
\int_A \mathrm{tr}(\omega'_{k,l}(x) M_{k,l}) d\mu(x)$, where 
$M_{k,l}=\mathrm{tr}_q (I \otimes |l\rangle \langle k | G)$ and 
$\omega'_{k,l} \in O'_h([\omega_{k,l}])$ with $\omega_{k,l} : x \mapsto
\mathrm{tr}_q (\Omega(x) I\otimes |k\rangle \langle l |)$. Let $w_1$ and 
$w_2$ be any simple probability measures of $hq$, $[\Omega_1]$ and 
$[\Omega_2]$ the corresponding classes in $\mathrm{P}'$, $t \in [0,1]$, and 
$w=t w_1+(1-t)w_2$. The convex-linearity of the bijection between the set of 
simple probability measures of $hq$ and $\mathrm{P}'$ implies that 
$\omega_{k,l}=t \omega_{1,k,l}+(1-t)\omega_{2,k,l}$ $\mu$-almost 
everywhere, where $\omega_{n,k,l}$ is defined similarly to $\omega_{k,l}$ 
with $\Omega$ replaced by $\Omega_n$. It follows from the linearity of $O'_h$ 
that $\omega'_{k,l}=t \omega'_{1,k,l}+(1-t)\omega'_{2,k,l}$ $\mu$-almost 
everywhere, where $\omega'_{n,k,l} \in O'_h([\omega_{n,k,l}])$. This equality, 
together with the above expression of $\cal O$, implies the convex-linearity of 
$\cal O$ for simple probability measures. 
\end{proof}
A requirement for bipartite systems consisting of $h$ and any classical system, 
say $c$, similar to that defining hybrid operations can be formulated as follows. 
Let $c$ be characterized by sample space $Y$, event space ${\cal B}$, and 
classical reference measure $\nu$. The bipartite system $hc$ consisting of $h$ 
and $c$ is a hybrid system whose events are the elements of the product 
$\sigma$-algebra ${\cal A} \otimes {\cal B}$, effects are the elements of 
$\cal E$, and classical reference measure is the product measure 
$\mu \otimes \nu$. For any probability measure $w$ of $hc$ and any event 
$B$ of $c$ such that the probability $p_c(B)=w(X \times B,I)$ is nonvanishing, 
a probability measure $w_{B}$ of $h$ can be defined by 
$w_B(A,E)=w(A \times B,E)/p_c(B)$. It can be interpreted similarly as $w_F$ 
given by eq.\eqref{wF}. We show below that there is a probability measure 
$w'$ of $hc$ such that, for any event $A$ and effect $E$ of $h$, 
$w'(A \times B,E)=p_c(B){\cal C}_h(w_B)(A,E)$ if $p_c(B)>0$ and vanishes 
otherwise, for any continuous convex-linear map ${\cal C}_h$ and discrete 
system $c$. Due to Lemma \ref{convlin}, this result applies to continuous 
operations. 
\begin{proof}
Consider the set of $\mu \otimes \nu$-Bochner integrable maps from 
$X \times Y$ to the set of trace-class operators on ${\cal H}$. We denote as 
$\mathrm{L}_{hc}$ the corresponding quotient set and as $\mathrm{P}_{hc}$ 
the subset of $\mathrm{L}_ {hc}$ consisting of all classes $[\Omega]$ 
containing a map $\Omega$ such that $\Omega(x,y)$ is positive for any 
$(x,y) \in X \times Y$ and 
$\int_{X \times Y} \mathrm{tr}\Omega(x,y) d\xi(x,y) =1$ where 
$\xi=\mu\otimes\nu$. For a discrete system $c$, $Y$ is a subset of 
$\mathbb{N}$, $\cal B$ is the set of all subsets of $Y$, and $\nu$ is the 
counting measure. Let $[\Omega] \in \mathrm{P}_{hc}$ be the class
corresponding to $w$ and define, for any $y \in Y$, 
$\omega_y : x \mapsto \Omega(x,y)$ on $X$. For any bounded operator $M$ 
on ${\cal H}$, $(x,y) \mapsto \mathrm{tr}(\Omega(x,y)M)$ is 
$\mu \otimes \nu$-integrable and so $x \mapsto \mathrm{tr}(\omega_y(x)M)$ 
is $\mu$-integrable. Fubini's theorem implies that 
$\sum_{y \in Y} {\cal I}_y =1$ where 
${\cal I}_y=\int_X \mathrm{tr}\omega_y(x) d\mu(x)$ and so ${\cal I}_y$ is 
finite for any $y \in Y$ \cite{Bi}. Consequently, due to Pettis measurability 
theorem, $\omega_y$ is $\mu$-Bochner integrable \cite{RN}. Let $C'_h$ be 
the bounded linear map from $\mathrm{L}_h$ to itself corresponding to 
${\cal C}_h$ and denote $[\omega'_y]=C'_h([\omega_y])$ where 
$\omega'_y(x)$ is positive for any $x \in X$ and 
$\vert [ \omega'_y ] \vert ={\cal I}_y$, see the proof of Lemma \ref{cont}. 
As $[\omega'_y] \in \mathrm{L}_h$, there is a sequence of simple functions 
$\hat \omega_{y,n}$ such that $\lim_{n \rightarrow \infty} 
\Vert \omega'_y(x)-\hat \omega_{y,n}(x) \Vert=0$ for any $x \in X$. 

We define $\Omega' : (x,y) \mapsto \omega'_y(x)$ and the simple functions 
$\hat \Omega_n$ on $X\times Y$ by $\hat \Omega_n(x,y)=\omega_{y,n}(x)$ 
if $y \le n$ and vanishes otherwise. For any $(x,y) \in X \times Y$, 
$\Omega'(x,y)$ is positive and $\lim_{n \rightarrow \infty} 
\Vert \Omega'(x,y)-\hat \Omega_{n}(x,y) \Vert=0$. Moreover, Tonelli's theorem 
gives $\int_{X \times Y} \mathrm{tr}\Omega'(x,y) d\xi(x,y) =1$ \cite{Bi}. So, 
$[\Omega']$ belongs to $\mathrm{P}_{hc}$. We denote as $w'$ the 
corresponding probability measure of $hc$. Using again Fubini's theorem, one 
obtains $w'(A\times B,E)
=\sum_{y \in B} \int_A \mathrm{tr} (\omega'_y(x)E) d\mu(x)$. Let us define 
$[\omega_B]=\sum_{y \in B} [\omega_y]/p_c(B)$, where the sum converges 
with respect to the norm $\vert \cdot \vert$, and 
$[\omega'_B]=C'_h([\omega_B])$. As $C'_h$ is a bounded linear map, one has 
$[\omega'_B]=\sum_{y \in B} [\omega'_y]/p_c(B)$. The required property of 
${\cal C}_h$ then follows from 
$w(A\times B,E)= p_c(B)\int_A \mathrm{tr} (\omega_B(x)E) d\mu(x)$ and the 
similar relation between $w'$ and $\omega'_B$.
\end{proof}
\subsection{Generalized Kraus theorem}
The following theorem can now be shown.
\begin{Kt}
For any hybrid operation ${\cal O}_h$, there is a sequence of maps 
$K_{\alpha}$ from ${\cal A}^2$ to the set of bounded operators on $\cal H$ 
such that, for any simple probability measure $\hat w$, given by 
eq.\eqref{wapp} with events $A_n$ and operators $\omega_n$, event $A$, 
and effect $E$,
\begin{equation}
{\cal O}_h(\hat w)(A,E) = \sum_{n,\alpha}\mathrm{tr}
\left(K_{\alpha}(A,A_n)\omega_n K_{\alpha}(A,A_n)^\dag E \right)  . 
\label{Kteq}
\end{equation}

Any map from the set of simple probability measures of $h$ to ${\cal P}_h$ 
given by eq.\eqref{Kteq} is the restriction of a unique continuous map on 
${\cal P}_h$ which is a hybrid operation.
\end{Kt}
As seen above, when ${\cal H}$ is finite-dimensional or the classical subsystem 
$c_h$ is discrete, all hybrid operations are continuous. Thus, in these cases, the 
above theorem implies that an operation is fully characterized by its action on 
simple probability measures and so by the Kraus maps $K_\alpha$ that appear 
in eq.\eqref{Kteq}. The transformed probability measure ${\cal O}_h(\hat w)$ is 
not necessarily simple. Moreover, though the existence of a corresponding hybrid 
state is ensured by Corollary \ref{Cor1}, it does not follow straightforwardly from 
eq.\eqref{Kteq}. We consider below specific operations and systems for which 
the transformed state can always be written in terms of the initial one, simple or 
not. If $X$ is a singleton, all events in eq.\eqref{Kteq} are equal to $X$ and so 
${\cal O}_h$ reduces to an usual quantum operation, i.e., 
${\cal O}_h(\hat w)(X,E) =\mathrm{tr}( \rho' E)$ with the transformed state 
$\rho'=\sum_{\alpha} L_{\alpha}^{\phantom{\dag}}\rho L_{\alpha}^\dag$ 
of the quantum subsystem $q_h$, where $\rho$ is the initial one and 
$L_{\alpha}=K_{\alpha}(X,X)$.
\begin{proof}
Lemma \ref{convlin} ensures that ${\cal O}_h$ is convex-linear. We denote as 
$O'_h$ the bounded linear map from $\mathrm{S}_h$ to $\mathrm{L}_h$ 
corresponding to ${\cal O}_h$, see the proof of Lemma \ref{cont}. Let $A$ and 
$B$ be any events of $h$ with finite $\mu(A)$ and define $O''_h$ from 
${\cal T}_h$ to itself by $O''_h(T)=\int_B \omega d\mu$, where 
$\omega \in  O'_h([T 1_A])$. Since $O'_h$ is linear, so is $O''_h$. Moreover, it 
obeys $\Vert O''_h(T) \Vert \le e \Vert T \Vert$ where $e=2 \mu(A)$, as the 
norm of $O'_h$ is at most $2$. Let $M$ be any bounded operator on $\cal H$ 
and $\xi$ the linear form on ${\cal T}_h$ defined by 
$\xi(T)=\mathrm{tr}(O''_h(T) M)$. The linearity of $\xi$ follows from that of 
$O''_h$. It is furthermore continuous as 
$|\xi(T)| \le e\Vert M \Vert_{op} \Vert T \Vert$. Therefore, there is a unique 
bounded operator $N$ such that $\xi(T)=\mathrm{tr}(T N)$ and 
$\Vert N \Vert_{op}=\sup_{T \ne 0} |\xi(T)|/\Vert T \Vert 
\le e\Vert M \Vert_{op}$ \cite{BS}. We define the dual map $O^*_h$ on the 
set of bounded operators by $O^*_h(M)=N$. It fulfills, for any elements 
$| k \rangle$ and $|l \rangle$ of an orthonormal basis of $\cal H$, 
$\langle k | O^*_h(M) | l \rangle =  
\mathrm{tr}(O''_h( | l \rangle \langle k |)M)$ and so is linear. The above
inequality shows that it is bounded. For any positive operator $M$ and 
$|\psi \rangle \in {\cal H}$, it results from the continuity of the linear form 
$T \mapsto \mathrm{tr}(TM)$ on ${\cal T}_h$ that 
$\langle \psi | O^*_h(M) |\psi \rangle = 
\int_B \mathrm{tr}(\omega(x)M) d\mu(x)$  where 
$\omega \in O'_h([|\psi \rangle \langle \psi | 1_A])$ \cite{RN}. For almost 
every $x$, $\omega(x)$ is positive, see the proof of Lemma \ref{cont}, and 
so $O^*_h(M)$ is positive. 

Consider the bipartite system $hq$ consisting of $h$ and any quantum system 
$q$ with finite-dimensional Hilbert space ${\cal H}_q$, the convex-linear map 
$\cal O$ given by Lemma \ref{Lem3}, and the corresponding bounded linear 
map $O' :\mathrm{S} \rightarrow \mathrm{L}$, see the proof of Lemma 
\ref{cont}. Linear maps $O''$ from $\cal T$ to itself, and $O^*$ from the set 
of bounded operators on ${\cal H} \otimes {\cal H}_q$ to itself can be defined 
similarly as above. The dual map $O^*$ is positivity-preserving. The equality of 
the right sides of eq.\eqref{eu1} and eq.\eqref{eu2}, the linearity of $O'_h$ 
and $O'$, and the continuity of the linear forms $T \mapsto \mathrm{tr} (TM)$ 
on $\cal T$ and ${\cal T}_h$, where $M$ is any bounded operator, give that, 
for any effects $E$ and $F$ of $h$ and $q$, respectively, and any rank-one 
projector $\Pi$, $\mathrm{tr}(O''(\Pi) E\otimes F)
=\mathrm{tr}(O''_h(\mathrm{tr}_q (\Pi I\otimes F)) E) $, which implies 
$O^*( E\otimes F)= O_h^*(E)\otimes F$. Any bounded operator can be 
decomposed into positive operators, see the proof of Lemma \ref{cont}. 
Moreover, for any positive operator $M$, there is an effect $E$ such that 
$M=\Vert \sqrt{M} \Vert_{op}^2 E$. Consequently, the above relation between 
the linear maps $O^*$ and $O_h^*$ holds for any bounded operators. Thus, 
as $O^*$ is a positive map, $O_h^*$ is a completely positive map \cite{cpm}. 
The $\sigma$-weak topology on the set of bounded operators on ${\cal H}$ is 
generated by the subsets $V$ consisting of operators $M$ such that 
$\mathrm{tr}(T M) \in U$, where $T$ is trace-class and $U$ is an open set of 
$\mathbb{C}$ \cite{BLT}. Denote as $V'$ the inverse image by $O_h^*$ of 
such a subbase set $V$. For any $M \in V'$, one has 
$\mathrm{tr}(O''_h(T) M)=\mathrm{tr}(T O_h^*(M)) \in U$, since 
$O_h^*(M) \in V$, and so $V'$ is open, which shows that $O_h^*$ is 
$\sigma$-weakly continuous.

Therefore, $O_h^*$ fulfills all the prerequisites of Kraus theorem \cite{K}. So, 
there is a sequence of maps $K_{\alpha}$ from ${\cal A}^2$ to the set of 
bounded operators on ${\cal H}$ such that, for any bounded operator $M$, 
$O_h^*(M)=\sum_{\alpha} K_{\alpha}(B,A)^{\dag} M K_{\alpha}(B,A)$, 
where the sum converges weakly. Let $\hat w$ be a simple probability measure 
of $h$. The corresponding class $[\hat \omega]$ belongs to $\mathrm{P}'_h$ 
and so one can write $\hat \omega=\sum_n \omega_n 1_{A_n}$, where the 
sum is finite and $\omega_n$ denotes positive trace-class operators. The 
linearity of $O'_h$ and the continuity of the linear form 
$T \mapsto \mathrm{tr}(TE)$ on ${\cal T}_h$ give ${\cal O}_h(\hat w)(A,E) 
= \sum_n \mathrm{tr}( \int_A \hat \omega_n d\mu E)$ where 
$\hat \omega_n \in O'_h([\omega_n 1_{A_n}])$.  Using the above results, the 
summand in this expression can be rewritten with the help of the maps 
$K_{\alpha}$. Equation \eqref{Kteq} follows from the fact that the maps 
$E \mapsto \mathrm{tr}(\omega_n E)$ from $\cal E$ to $\mathbb{C}$ 
are weakly sequentially continuous, as $\omega_n$ is trace-class and positive, 
see the proof of the generalized Gleason theorem. 

Consider now any map ${\cal O}_h$ from the set of simple probability measures 
of $h$ to ${\cal P}_h$ given by eq.\eqref{Kteq}. Let $hq$ be the bipartite 
system consisting of $h$ and any quantum system $q$ with finite-dimensional 
Hilbert space ${\cal H}_q$ and $A$ and $B$ be any events with finite $\mu(A)$. 
For any $|\psi \rangle$ and $|\phi \rangle$ of ${\cal H}\otimes{\cal H}_q$, 
one has  
$\langle \psi | L_{\alpha}^{\dag} L_{\alpha} \otimes I_q | \phi \rangle 
= \sum_{u=1}^4 z_u \mathrm{tr} ( L_{\alpha} 
\mathrm{tr}_q (| \xi_u \rangle \langle \xi_u |)L_{\alpha}^{\dag})/2$, 
where $L_{\alpha}=K_{\alpha}(B,A)$, $| \xi_1 \rangle=| \psi \rangle$, 
$| \xi_2 \rangle=| \phi \rangle$, 
$| \xi_3 \rangle=| \psi \rangle+| \phi \rangle$,
$| \xi_4 \rangle=| \psi \rangle+i| \phi \rangle$,
$z_1=z_2=i-1$, $z_3=1$, and $z_4=-i$ and so the sum 
$\sum_{\alpha} L_{\alpha}^{\dag} L_{\alpha} \otimes I_q$ converges weakly. 
Thus, $\sum_{\alpha} L_{\alpha} \otimes I_q T L_{\alpha}^{\dag} \otimes I_q$ 
converges in $\cal T$ for any trace-class operator $T$ \cite{K}. Consequently, 
we can define a map $\cal O$ from the set of simple probability measures of 
$hq$ to ${\cal P}_{hq}$ by 
\begin{equation}
{\cal O}(\hat w)(A,G) = \sum_n 
\mathrm{tr}\Big(\sum_{\alpha}L_{\alpha,n} \otimes I_q 
\Omega_n L_{\alpha,n}^\dag \otimes I_q G \Big)  , \nonumber
\end{equation}
where $A$ and $G$ are any event and effect, respectively, $\hat w$ is a simple 
probability measure with events $A_n$ and positive trace-class operators 
$\Omega_n$, and $L_{\alpha,n}=K_{\alpha}(A,A_n)$.

Since the operators $\Omega'_n=\sum_{\alpha}L_{\alpha,n} \otimes I_q 
\Omega_n L_{\alpha,n}^\dag \otimes I_q$ are trace-class and positive, 
${\cal O}(\hat w)$ satisfies conditions \ref{i} and \ref{iii}. Conditions \ref{iv} 
and \ref{v} are fulfilled due to the continuity of the linear form 
$T \mapsto \mathrm{tr}T$ on $\cal T$, 
$[\sum_n \mathrm{tr}_q \Omega_n 1_{A_n}] \in \mathrm{P}_h$, and 
${\cal O}(\hat w)(A,G) \le {\cal O}(\hat w)(A,I\otimes I_q)$, which follows 
from condition \ref{iii}. To show that condition \ref{ii} is met, consider a 
sequence of pairwise disjoint events $B_m$ and $A=\cup_m B_m$ and define 
$\Omega'_{n,m}=\sum_{\alpha} L_{\alpha,m,n} \otimes I_q 
\Omega_n L_{\alpha,m,n}^\dag \otimes I_q$ with 
$L_{\alpha,m,n}=K_{\alpha}(B_m,A_n)$. The positivity of $\Omega'_{n,m}$ 
implies that $$\Big\Vert \sum_{m\le m'} \Omega'_{n,m} - 
\sum_{m\le m''} \Omega'_{n,m} \Big\Vert=\Big\vert \sum_{m\le m'} s_{n,m}  
- \sum_{m\le m''} s_{n,m} \Big\vert , $$ where 
$s_{n,m}=\sum_{\alpha}\mathrm{tr}
(L_{\alpha,m,n} \mathrm{tr}_q(\Omega_n) L_{\alpha,m,n}^\dag)$. 
As the codomain of ${\cal O}_h$ is ${\cal P}_h$, the sum $\sum_m s_{n,m}$ 
converges and so the above distance vanishes as $m'$ and $m''$ go to infinity. 
Thus, as $\cal T$ is a Banach space \cite{BS}, 
$\sum_m \Omega'_{n,m}$ converges to a trace-class operator $\Omega''_n$. 
For any $| k , k' \rangle=| k  \rangle \otimes | k' \rangle$ and 
$| l , l' \rangle=| l  \rangle \otimes | l' \rangle$, where $| k \rangle$ 
($| k' \rangle$) and $| l \rangle$ ($| l' \rangle$) are elements of an orthonormal 
basis of $\cal H$ (${\cal H}_q$), one can write 
$$\langle k , k' | \Omega'_n | l , l' \rangle = \sum_{u,v,\alpha}z_u z_v 
\mathrm{tr}\big(L_{\alpha,n} \omega_{n,k',l',u} 
L_{\alpha,n}^\dag E_{k,l,v}  \big), $$
where $(u,v)$ runs over $\{1,2,3\}^2$, 
$E_{k,l,1}=(| k \rangle+| l \rangle)(\langle k |+\langle l |)/2$, 
$E_{k,l,2}=(| k \rangle+i| l \rangle)(\langle k |-i\langle l |)/2$, 
and $E_{k,l,3}=(|k \rangle\langle k| +|l \rangle\langle l|)/2$ are 
effects, $\omega_{n,k',l',u}=\mathrm{tr}_q(I\otimes E_{k',l',u} \Omega_n)$ 
is a positive trace-class operator, $z_1=1$, $z_2=-i$, and $z_3=i-1$, and 
$\langle k , k' | \Omega'_{n,m}| l , l' \rangle$ is given by the same expression 
with $L_{\alpha,n}$ replaced by $L_{\alpha,m,n}$. Thus, as the codomain of 
${\cal O}_h$ is ${\cal P}_h$, one has $\langle k , k' | \Omega'_n | l , l' \rangle
= \sum_m \langle k , k' | \Omega'_{n,m}| l , l' \rangle$. This sum is also equal 
to $\langle k , k' | \Omega''_n | l , l' \rangle$ since 
$\Vert | l , l' \rangle \langle k , k' | \Vert_{op}=1$. So, 
$\Omega'_n=\Omega''_n$ is the limit in trace norm of 
$(\sum_{m\le m'} \Omega'_{n,m})_{m'}$. It follows that ${\cal O}(\hat w)$ 
satisfies condition \ref{ii} and is hence a probability measure.

As $\hat w(A,E\otimes F) = \sum_n \mathrm{tr}( \tau_n E )\mu(A \cap A_n)$, 
with the positive trace-class operators 
$\tau_n=\mathrm{tr}_q(\Omega_n I \otimes F)$, the conditional probability 
measure ${\hat w}_F$ of $h$, given by eq.\eqref{wF}, is simple. Thus, 
${\cal O}_h({\hat w}_F)$ is given by the right side of eq.\eqref{Kteq} with 
$\tau_n/\hat w(X,I\otimes F)$ in place of $\omega_n$. Since 
${\cal O}(\hat w)(A,E\otimes F) = \sum_{n,\alpha} 
\mathrm{tr}(L_{\alpha,n} \tau_n L_{\alpha,n}^\dag E )$, equation 
\eqref{vieq} is fulfilled with ${\cal O}(\hat w)$ in place of $w'$ and $\hat w$ 
in place of $w$. For any simple probability measures $\hat w_1$ and 
$\hat w_2$, there is a finite sequence of pairwise disjoint events $A_n$ and 
positive trace-class operators $\Omega_{k,n}$ such that 
$\hat w_k(A,G)=\sum_n \mathrm{tr}(\Omega_{k,n} G)\mu(A \cap A_n)$ for 
any event $A$ and effect $G$. Then, it follows from the definition of the map 
${\cal O}$ that ${\cal O}(t\hat w_1+(1-t)\hat w_2)
=t{\cal O}(\hat w_1)+(1-t){\cal O}(\hat w_2)$, where $t \in [0,1]$, i.e., 
${\cal O}$ is convex-linear. Equation \eqref{Kteq} leads to the same conclusion 
for ${\cal O}_h$. We denote as ${\cal O}'$ (${\cal O}'_h$) the unique 
continuous extension of ${\cal O}$ (${\cal O}_h$) given by Lemma \ref{cont} 
and as $O'$ ($O'_h$) the corresponding map from $\mathrm{L}$ 
($\mathrm{L}_h$) to itself. 

Consider any $[\Omega]\in\mathrm{P}$ and $[\Theta]\in\mathrm{P}'$. As 
shown above, equation \eqref{vieq} holds for any simple probability measure 
of $hq$ and so $\int_A \mathrm{tr}(\Theta'(x) E\otimes F) d\mu(x)
= \int_A \mathrm{tr}(\theta'(x) E) d\mu(x)$ where $\Theta' \in O'([\Theta])$ 
and $\theta' \in O'_h([\theta])$ with 
$\theta : x \mapsto\mathrm{tr}_q (\Theta(x) I\otimes F)$. Consequently, 
one has 
\begin{multline}
\left|\int_A \mathrm{tr}(\Omega'(x) E\otimes F) d\mu(x)
-\int_A \mathrm{tr}(\omega'(x) E) d\mu(x) \right| \\
\le \Vert G \Vert_{op} |[\Omega-\Theta]| 
+\Vert E \Vert_{op} |[\omega-\theta]| 
\le 2 |[\Omega-\Theta]| ,
\nonumber
\end{multline} 
where $G=E \otimes F$, $\Omega' \in O'([\Omega])$, and 
$\omega' \in O'_h([\omega])$ with 
$\omega : x \mapsto \mathrm{tr}_q (\Omega(x) I\otimes F)$. The above 
inequalities are obtained using the fact that the norms of $O'_h$ and $O'$ 
on $\mathrm{H}_h$ and $\mathrm{H}$, respectively, are not larger than unity, 
$\Vert \mathrm{tr}_q T \Vert \le \Vert T \Vert$ for any self-adjoint trace-class 
operator $T$ on ${\cal H} \otimes {\cal H}_q$, and 
$\Vert G \Vert_{op}=\Vert E \Vert_{op}\Vert F \Vert_{op} \le 1$ \cite{BS}. 
Corollary \ref{Corapp} implies that 
$\mathrm{P} \subset \overline{\mathrm{P}'}$ 
and so the above difference vanishes, which shows that ${\cal O}'_h$ is a 
hybrid operation.
\end{proof}
For continuous operations, it can be shown, using the above theorem, that the
transformation ${\cal O}$ given by Lemma \ref{Lem3}, that describes the 
action of ${\cal O}_h$ on $h$ in the presence of a quantum system $q$ with 
no interaction beween $h$ and $q$ and no intrinsic evolution of $q$, is an 
operation. Such an operation also exists when the Hilbert space of $q$ is 
infinite-dimensional. Furthermore, another characterization of continuous 
operations can be derived, see the following corollary.
\begin{Cor}\label{CorO}
Let ${\cal O}_h$ be a continuous map from ${\cal P}_h$ to itself. The following 
assertions are equivalent.
\begin{enumerate}[label=(\alph*)]
\item ${\cal O}_h$ is an operation, \label{1}
\item For any quantum system $q$, there is a continuous operation ${\cal O}$ 
of $hq$ such that, for any probability measure $w$ of $hq$, $w'={\cal O}(w)$ 
satisfies eq.\eqref{vieq}, \label{2}
\item For any quantum system $q$, there is a continuous operation ${\cal O}$ 
of $hq$ that transforms any product state 
$x \mapsto \omega_h(x) \otimes \rho_q$, where $\omega_h$ and $\rho_q$ 
denote states of $h$ and $q$, respectively, into 
$x \mapsto \omega'_h(x) \otimes \rho_q$, where $\omega'_h$ is the image of 
$\omega_h$ by ${\cal O}_h$. \label{3}
\end{enumerate}
For both assertions \ref{2} and \ref{3}, ${\cal O}$ is given by
\begin{equation}
{\cal O}(\hat w)(A,G) \!=\!\! \sum_{n,\alpha}\!
\mathrm{tr}\big(K_{\alpha}(A,A_n) \otimes I_q \Omega_n 
K_{\alpha}(A,A_n)^\dag \otimes I_q G \big)  , \label{OCor}
\end{equation}
for any simple probability measure $\hat w$ with events $A_n$ and operators 
$\Omega_n$, event $A$ and effect $G$ of $hq$, where $K_{\alpha}$ denotes 
the Kraus maps of ${\cal O}_h$ appearing in eq.\eqref{Kteq}.
\end{Cor}
\begin{proof} 
Assume that ${\cal O}_h$ is an operation. It follows from the above theorem 
that it transforms simple probability measures according to eq.\eqref{Kteq}. 
Consider the map from the set of simple probability measures of $hq$ to 
${\cal P}_{hq}$ given by eq.\eqref{OCor}, see the above proof. Note that the 
arguments used hold also when the Hilbert space of $q$ is infinite-dimensional. 
As eq.\eqref{OCor} is a particular case of eq.\eqref{Kteq}, the above theorem 
ensures that there is a unique continuous map ${\cal O}$ on ${\cal P}_{hq}$ 
that fulfills eq.\eqref{OCor} and that it is an operation. As shown in the proof 
of this theorem, $w'={\cal O}(w)$ satisfies eq.\eqref{vieq}. 

Consider the class $[\hat \omega \otimes \rho]$ where $[\hat \omega]$ is any 
class of $\mathrm{P}'_h$ and $\rho$ is any density operator of $q$ and denote 
as $\hat w$ the corresponding simple probability measure of $hq$. For any 
event $A$ and effect $G$ of $hq$, using eq.\eqref{OCor} and noting that 
$\mathrm{tr}_q(I \otimes \rho G)$ is an effect of $h$, one can write 
${\cal O}(\hat w)(A,G) = 
\mathrm{tr}(\int_A \hat \omega' d\mu \otimes \rho G)$ 
with $\hat \omega' \in O'_h([\hat \omega])$ where $O'_h$ is the bounded 
linear map from $\mathrm{L}_h$ to itself corresponding to ${\cal O}_h$, see 
the above proofs. So, as $T \mapsto T \otimes \rho$ is a bounded linear map 
from ${\cal T}_h$ to ${\cal T}$, the class corresponding to ${\cal O}(\hat w)$ 
is $[\hat \omega' \otimes \rho]$, see eq.\eqref{Gtmu}. Consider now 
$[\omega \otimes \rho]$ where $[\omega]$ is any class of $\mathrm{P}_h$ 
and denote as $w$ the corresponding probability measure of $hq$ and as $w'$ 
the probability measure given by 
$w'(A,G) =\int_A \mathrm{tr}(\omega'(x) \otimes \rho G) d\mu(x)$ with 
$\omega' \in O'_h([\omega])$. It follows from Corollary \ref{Corapp} that there 
is a sequence of classes $[\hat \omega_{n}]$ in $\mathrm{P}'_h$ such that 
$|[\omega - \hat \omega_{n}]|$ vanishes as $n \rightarrow \infty$. The 
continuity of ${\cal O}$ and $O'_h$ give 
$\lim_{n \rightarrow \infty} d({\cal O}(w),{\cal O}(\hat w_n))=0$, where 
$\hat w_n$ denotes the probability measure corresponding to 
$[\hat \omega_{n} \otimes \rho]$, and 
$\lim_{n \rightarrow \infty} |[\omega'- \hat \omega'_{n}]|=0$, where 
$\hat \omega'_{n} \in O'_h([\hat \omega_{n}])$. One has 
$d({\cal O}(w),w') \le d({\cal O}(w),{\cal O}(\hat w_n))+
|[\omega' - \hat \omega'_{n}]|$ for any $n$ and so ${\cal O}(w)=w'$, i.e., 
the class corresponding to ${\cal O}(w)$ is $[\omega' \otimes \rho]$.

We have shown that \ref{1} implies both \ref{2} and \ref{3}. It is obvious that 
\ref{2} leads to \ref{1}. It remains to show that \ref{3} gives \ref{1}. Assume 
\ref{3} holds, consider a quantum system $q$ with finite-dimensional Hilbert 
space ${\cal H}_q$, and denote as $O'$ the bounded linear map from 
$\mathrm{L}$ to itself corresponding to $\cal O$, see the above proofs. 
Assertion \ref{3} means that, for any $[\omega] \in \mathrm{P}_h$ and 
density operator $\rho$ of $q$, 
$O'([\omega \otimes \rho])=[\omega' \otimes \rho]$, where 
$\omega' \in O_h([\omega])$ with the map $O_h$ from $\mathrm{P}_h$ to 
itself corresponding to ${\cal O}_h$. Let $w$ be any probability measure of 
$hq$, $[\Omega]$ be the corresponding class of $\mathrm{P}$, and write, 
for any $x$, $\Omega(x)=\sum_{k,l} \sum_{u,v} z_u z^*_v 
\tilde \omega_{k,l,u}(x) \otimes \rho_{k,l,v}$, where the first sum runs over all 
couples of elements $|k\rangle$ and $|l\rangle$ from an orthonormal basis of 
${\cal H}_q$, $(u,v)$ runs over $\{ 1,2,3 \}^2$, 
$\rho_{k,l,1}=(| k \rangle+| l \rangle)(\langle k |+\langle l |)/2$, 
$\rho_{k,l,2}=(| k \rangle+i| l \rangle)(\langle k |-i\langle l |)/2$, 
and $\rho_{k,l,3}=(|k \rangle\langle k| +|l \rangle\langle l|)/2$ are density 
operators of $q$, 
$\tilde \omega_{k,l,u}(x)=\mathrm{tr}_q(I\otimes \rho_{k,l,u} \Omega(x))$ 
is trace-class and positive, $z_1=1$, $z_2=-i$, and $z_3=i-1$. Assertion \ref{3} 
leads to $O'([\Omega])=\sum_{k,l,u,v} z_u z^*_v t_{k,l,u}
[\omega'_{k,l,u} \otimes \rho_{k,l,v}]$ where 
$t_{k,l,u}=\int \mathrm{tr} \tilde \omega_{k,l,u}(x) d\mu(x)$ and 
$\omega'_{k,l,u} \in O_h([\omega_{k,l,u}])$ with 
$\omega_{k,l,u}=\tilde \omega_{k,l,u}/t_{k,l,u}$ if $t_{k,l,u}>0$ and vanishes 
otherwise. For any effect $F$ of $q$ such that $w(X,I\otimes F)>0$, the class 
of $\mathrm{P}_h$ corresponding to the probability measure given by 
eq.\eqref{wF} is $\sum_{k,l,u,v} z_u z^*_v t_{k,l,u} 
\mathrm{tr}(\rho_{k,l,v} F) [\omega_{k,l,u}]/w(X,I\otimes F)$. So, 
$w'={\cal O}(w)$ satisfies eq.\eqref{vieq} and \ref{1} is fulfilled.
\end{proof}
\section{Specific cases}\label{So}
\subsection{Non-interacting classical and quantum subsystems}\label{Icqs}
Consider hybrid probability measures of the form 
$(A,E) \mapsto p(A)\mathrm{tr}(\rho E)$, where $p$ is the probability measure 
of the classical subsystem $c_h$ and $\rho$ is the density operator of the 
quantum subsystem $q_h$. Under the action of a hybrid operation describing
non-interacting subsystems, such a product probability measure is transformed 
into another one given by 
$(A,E) \mapsto {\cal O}_c(p)(A) \mathrm{tr}(O_q(\rho)E)$, where 
${\cal O}_c$ is a probability measure transformation of $c_h$ and $O_q$ is a 
map from the set of density operators of $q_h$ to itself. For such operations, 
the following result can be shown. 
\begin{Prop}\label{Prop1}
Let ${\cal O}_h$ be a map from ${\cal P}_h$ to itself. 

The map ${\cal O}_h$ is a continuous operation describing non-interacting 
classical and quantum subsystems if and only if there are a map 
$k : {\cal A}\times X \rightarrow \mathbb{R}^+$ and a sequence of bounded 
operators $L_{\alpha}$ on ${\cal H}$ such that $k(X,x)=1$ for any $x$, 
$\sum_{\alpha} L_{\alpha}^\dag L_{\alpha}=I$, and, for any probability 
measure $w$, 
\begin{equation}
{\cal O}_h(w)(A,E) = \sum_{\alpha} \int k(A,x)
\mathrm{tr}\left(L_{\alpha}\omega(x) L_{\alpha}^\dag E \right)  d\mu(x), 
\label{indop}
\end{equation}
where $A$ and $E$ are any event and effect, respectively, and $\omega$ is 
given by Corollary \ref{Cor1}. 
\end{Prop}
The transformed probability measure of $c_h$ is 
$A \mapsto \int k(A,x) f(x) d\mu(x)$, where 
$f : x \mapsto \mathrm{tr} \omega(x)$ is the initial probability density function 
of $c_h$. The transformed state of $q_h$ is 
$\rho'=\sum_{\alpha} L_{\alpha}\rho L_{\alpha}^\dag$, where 
$\rho=\int \omega d\mu$ is the initial state of $q_h$. In other words, $\rho$ 
is changed into $\rho'$ according to the quantum operation with Kraus 
operators $L_{\alpha}$. Equation \eqref{indop} is the expression of a hybrid 
probability measure when the map $k$ is a Markov kernel, i.e., 
$A \mapsto k(A,x)$ is a probability measure for any $x$ \cite{MK}, such that 
$x \mapsto k(A,x)$ vanishes almost everywhere when $A$ is null. Examples are 
$k(A,x)=\int_A h(x',x) d\mu(x')$, where $h$ is a $\mu \otimes \mu$-integrable 
map from $X\times X$ to $\mathbb{R}^+$, and $k(A,x)=1_A(\phi(x))$, where 
$\phi$ is a measurable function from $X$ to itself such that 
$\mu(\phi^{-1}(A))=0$ for any null event $A$. In the first case, $c_h$ evolves 
according to a stochastic dynamics and the transformed hybrid state can be 
written as $\omega'(x)=\sum_{\alpha} \int h(x,x')
L_{\alpha}\omega(x') L_{\alpha}^\dag  d\mu(x')$. In the second case, the 
dynamics of $c_h$ is deterministic, governed by the map $\phi$. Important 
transformations $\phi$ are the measure-preserving ones, i.e., such that 
$\mu(\phi^{-1}(A))=\mu(A)$ for any event $A$, for instance, those following 
from classical Liouville equations \cite{LE}. They obviously fulfill the above 
mentioned condition. When $\phi$ is, furthermore, invertible with measurable 
inverse, the transformed hybrid state can be written as 
$\omega'(x)=\sum_{\alpha} 
L_{\alpha}\omega(\phi^{-1}(x)) L_{\alpha}^\dag$ \cite{Bi}.
\begin{proof}
Assume ${\cal O}_h$ is a continuous operation describing non-interacting 
classical and quantum subsystems. Consider any probability measure of the 
form $\hat w : (A,E) \mapsto \mathrm{tr}(\rho E) \mu(A\cap B) / \mu(B)$, 
where $\rho$ is any density operator and $B$ is any event such that $\mu(B)$ 
is finite and non-vanishing. The above theorem gives 
${\cal O}_h(\hat w)(A,E) = \sum_{\alpha}
\mathrm{tr}(K_{\alpha}(A,B)\rho K_{\alpha}(A,B)^\dag E) / \mu(B)$ 
for any event $A$ and effect $E$. By assumption, these probabilities can also be 
written as ${\cal O}_h(\hat w)(A,E)=p_B(A) \mathrm{tr}(O_q(\rho) E)$, where 
$p_B$ is a classical probability measure depending on $B$ and $O_q$ is a map 
from the set of density operators of $q_h$ to itself which does not depend on 
$A$ and $B$. It follows that $p_B(A) O_q(\rho)=\sum_{\alpha} 
K_{\alpha}(A,B)\rho K_{\alpha}(A,B)^\dag / \mu(B)$, in particular, 
$O_q(\rho)=\sum_{\alpha} L_{\alpha}\rho L_{\alpha}^\dag$, where 
$L_{\alpha}=K_{\alpha}(X,C)/\sqrt{\mu(C)}$ with any event $C$ such that 
$\mu(C)$ is finite and non-vanishing. So, one has 
$\sum_{\alpha} K_{\alpha}(A,B)\rho K_{\alpha}(A,B)^\dag 
= \nu_A(B) \sum_{\alpha} L_{\alpha}\rho L_{\alpha}^\dag$, where 
$\nu_A(B)=p_B(A)\mu(B)$. This equality also holds for a $\mu$-null event $B$ 
with the definition $\nu_A(B)=0$, which also ensures $\nu_A \le \mu$. For 
infinite $\mu(B)$, we set $\nu_A(B)=\infty$. Since $\mathrm{tr}O_q(\rho)=1$ 
for any density operator $\rho$, the operators $L_{\alpha}$ obey 
$\sum_{\alpha} L_{\alpha}^\dag L_{\alpha}=I$.

Let $A_1$ and $A_2$ be any disjoint events and $B=A_1 \cup A_2$. It results 
from eq.\eqref{Kteq} and the above that $\nu_A(B)=\nu_A(A_1)+\nu_A(A_2)$, 
i.e., $\nu_A$ is finitely additive. Consider any infinite sequence of pairwise 
disjoint events $A_n$ and $B=\cup_n A_n$. One has 
$\nu_A(B) -\sum_{n<N} \nu_A(A_n)
=\nu_A(\cup_{n \ge N} A_n) \le \mu(B) -\sum_{n<N} \mu(A_n)$. This last 
difference vanishes as $N$ goes to infinity and so $\nu_A$ is a measure on 
$\cal A$. By construction, $\nu_A$ is absolutely continuous with respect to 
$\mu$. Since $\mu$ is $\sigma$-finite and $\nu_A \le \mu$, $\nu_A$ is 
$\sigma$-finite. Consequently, due to Radon-Nikodym theorem, there is a 
measurable positive function $k_A$ such that $\nu_A(B)=\int_B k_A d\mu$ for 
any $B \in {\cal A}$ \cite{Bi}. Since $\nu_X=\mu$, this equality is satisfied with 
$k_X=1_X$ for $A=X$. We define the map 
$k : {\cal A}\times X \rightarrow \mathbb{R}^+$ by $k(A,x)=k_A(x)$. It 
follows from the above results and eq.\eqref{Kteq} that eq.\eqref{indop} is 
fulfilled by any simple probability measure. 
  
Consider any probability measure $w$ and the map $w'$ on 
${\cal A}\times{\cal E}$ given by the right side of eq.\eqref{indop} with the 
class $[\omega]$ corresponding to $w$. Using the monotone convergence 
theorem and the fact that the linear form $T \mapsto \mathrm{tr}(T E)$ is 
continuous on ${\cal T}_h$, one can write 
$w'(A,E) = \int k(A,x)\mathrm{tr}\big(O'_q(\omega(x) ) E \big) d\mu(x)$ 
where the linear map $O'_q$ is given by 
$O'_q(T)=\sum_{\alpha} L_{\alpha}T L_{\alpha}^\dag$ for any trace-class 
operator $T$. If $T$ is positive, $O'_q(T)$ is positive and 
$\Vert O'_q(T)\Vert = \Vert T\Vert$. Any self-adjoint trace-class operator $T$ 
can be expanded as $T=T_+-T_-$ where $T_+$ and $T_-$ are positive 
trace-class operators such that their product vanishes \cite{BS}. These 
operators fulfill $\Vert T_+\Vert+\Vert T_-\Vert=\Vert T\Vert$ and so 
$\Vert O'_q(T)\Vert \le \Vert T\Vert$. Corollary \ref{Corapp} implies that, for 
any $\epsilon>0$, there is a simple probability measure $\hat w$ such that 
$d(w,\hat w)<\epsilon$. For any event $A$ and effect $E$, one has 
$|{\cal O}_h(w)(A,E)-{\cal O}_h(\hat w)(A,E)|<\epsilon$, since the continuous 
operation ${\cal O}_h$ obeys inequality \eqref{contineq} and 
$\Vert E \Vert_{op} \le 1$ \cite{BS}, see equations \eqref{Gtmu} and \eqref{d}. 
This leads to $|{\cal O}_h(w)(A,E)-w'(A,E)|<\epsilon+\int k(A,x)
\Vert O'_q(\omega(x)-\hat \omega(x)) \Vert  d\mu(x)$, where $[\hat \omega]$ 
is the class corresponding to $\hat w$. It follows from $\nu_A \le \mu$ that 
$\int_B (1-k(A,x)) d\mu(x) \ge 0$ for any event $B$ and so $k(A,x) \le 1$ for 
$\mu$-almost every $x$. This gives, together with the above proved property of 
$O'_q$, $|{\cal O}_h(w)(A,E)-w'(A,E)|<\epsilon+d(w,\hat w)<2\epsilon$ and 
hence ${\cal O}_h(w)=w'$.

Assume now that the map ${\cal O}_h$ is given by eq.\eqref{indop}. It 
transforms simple probability measures according to eq.\eqref{Kteq} with the 
Kraus maps 
$K_{\alpha} : (A,B) \mapsto [\int_B k(A,x) d\mu (x)]^{1/2} L_{\alpha}$. The 
above theorem implies that there is a continuous operation ${\cal O}'_h$ that 
coincides with ${\cal O}_h$ for simple probability measures. Since $\mu$ is 
$\sigma$-finite and, for any event  $B$ such that $\mu(B)$ is finite and 
non-vanishing, $A \mapsto \int_B k(A,x) d\mu (x)/\mu(B)$ is a probability 
measure on ${\cal A}$, one has $k(A,x) \le 1$ for $\mu$-almost every $x$. 
Then, arguing as above, it can be shown that ${\cal O}_h={\cal O}'_h$ and so 
${\cal O}_h$ is a continuous operation. Consider 
$w : (A,E) \mapsto p(A)\mathrm{tr}(\rho E)$ where $p$ is the probability 
measure of $c_h$ and $\rho$ is the density operator of $q_h$. Corollary 
\ref{Cor1} gives $[\omega]=[f \rho]$ where $f$ is a probability density function 
corresponding to $p$. Using eq.\eqref{indop}, one finds 
${\cal O}_h(w)(A,E)=p'(A)\mathrm{tr}(\rho' E)$ with the classical probability 
measure $p': A \mapsto \int k(A,x) f(x) d\mu(x)$ and the density operator 
$\rho'=\sum_{\alpha} L_{\alpha}\rho L_{\alpha}^\dag$ and so ${\cal O}_h$ 
describes non-interacting subsystems. 
\end{proof}
\subsection{Quantum subsystem with finite-dimensional Hilbert space}
For a finite-dimensional Hilbert space $\cal H$, the set of operators on $\cal H$ 
has finite bases and the following result can be shown.
\begin{Prop}
Let $\cal H$ be finite-dimensional and ${\cal O}_h$ be a map from ${\cal P}_h$ 
to itself.

The map ${\cal O}_h$ is an operation if and only if it is continuous and, for any 
basis $\{ L_{\alpha} \}_{\alpha}$ of the set of operators on ${\cal H}$, there 
are maps $k_{\alpha,\beta} : {\cal A}\times X \rightarrow \mathbb{C}$ 
such that, for any probability measure $w$, 
\begin{equation}
{\cal O}_h(w)(A,E) = \sum_{\alpha,\beta} \int k_{\alpha,\beta}(A,x)
\mathrm{tr}\big(L_{\alpha}\omega(x) L_{\beta}^\dag E \big)  d\mu(x), 
\label{fHs}
\end{equation}
where $A$ and $E$ are any event and effect, respectively, and $\omega$ is 
given by Corollary \ref{Cor1}. 
\end{Prop}
Possible examples for the maps $k_{\alpha,\beta}$ are transition kernels similar 
to those considered in subsection \ref{Icqs}. An alternative derivation of 
eq.\eqref{fHs} with 
$k_{\alpha,\beta}(A,x)=\int_A h_{\alpha,\beta}(x',x) d\mu(x')$, where 
$h_{\alpha,\beta}$ are $\mu \otimes \mu$-integrable maps from $X\times X$ 
to $\mathbb{R}^+$, can be found in Ref.\cite{OSSW}. As a specific 
example, assume that $\cal H$ is two-dimensional with orthonormal basis 
$\{ | + \rangle , | - \rangle \}$ and consider the operation that changes any 
hybrid state $\omega$ into 
$x \mapsto  \Pi_+\int h(x,x') f_+(x') d\mu(x')+\Pi_- f_- (\phi^{-1}(x))$, 
where $\Pi_{\pm}=| \pm \rangle \langle \pm |$, 
$f_{\pm}:x\mapsto \langle \pm | \omega(x) | \pm \rangle$, $h$ is a 
$\mu \otimes \mu$-integrable map from $X\times X$ to $\mathbb{R}^+$ 
such that $\int h(x,x') d\mu(x)=1$, and $\phi$ is a measure-preserving map 
from $X$ to itself with measurable inverse. In this case, the classical probability 
measures $p_{\pm} : A \mapsto w(A,\Pi_{\pm})/w(X,\Pi_{\pm})$ are 
transformed independently of each other. Moreover, the evolution of $p_+$ is 
stochastic, whereas that of $p_-$ is deterministic.
\begin{proof}
Assume ${\cal O}_h$ is an operation, with Kraus maps $K_\gamma$ and 
consider any orthonormal basis $\{ | k \rangle \}_k$ of ${\cal H}$. A basis of 
the set of operators on ${\cal H}$ is $\{ L_{\alpha} \}_{\alpha}$ where 
$\alpha=(k,l)$ and $L_{\alpha}=|k\rangle \langle l|$. Expanding the operators 
$K_{\gamma}(A,B)$, where $A$ and $B$ are any events such that $\mu(B)$ is 
finite, in terms of the basis operators $L_\alpha$ leads to 
$ \sum_{\gamma} K_{\gamma}(A,B) M K_{\gamma}(A,B)^\dag
= \sum_{\alpha,\beta}\nu_{\alpha,\beta,A}(B) L_{\alpha}M L_{\beta}^\dag$ 
for any operator $M$. The maps $\nu_{\alpha,\beta,A}$ are given by 
$\nu_{\alpha,\beta,A}(B)=\sum_{\gamma} 
\langle k| K_{\gamma}(A,B) |l\rangle \langle l'| 
K_{\gamma}(A,B)^\dag |k'\rangle$ where $\alpha=(k,l)$ and $\beta=(k',l')$.
 Using $| l \rangle \langle l' |=|\psi \rangle \langle \psi| 
+i|\phi \rangle \langle \phi| -(1+i)(|l \rangle\langle l| +|l' \rangle\langle l'|)/2$,
where $|\psi \rangle=(| l \rangle+| l' \rangle)/\sqrt{2}$ and 
$|\phi \rangle=(| l \rangle+i| l' \rangle)/\sqrt{2}$, and a similar expression for 
$|k'\rangle \langle k|$, one can write 
$\nu_{\alpha,\beta,A}=\sum_{n=1}^4 i^n \nu_{\alpha,\beta,A,n}$ 
where the maps $\nu_{\alpha,\beta,A,n}$ are positive and dominated by 
$3\mu$. They vanish for a $\mu$-null event $B$ and, for finite and 
nonvanishing $\mu(B)$, $\nu_{\alpha,\beta,A,n}(B)/\mu(B)$ is a sum of 
transformed probabilities ${\cal O}_h(w)(A,E)$, where 
$w : (A',E') \mapsto \mathrm{tr}(\rho E') \mu(A'\cap B) / \mu(B)$, with 
appropriate density operator $\rho$ and event $E$. We set 
$\nu_{\alpha,\beta,A,n}(B)=\infty$ for infinite $\mu(B)$. 

Following the same arguments as in the proof of Proposition \ref{Prop1}, it can 
be shown that $\nu_{\alpha,\beta,A,n}$ is a measure and that there is a 
measurable positive function $k_{\alpha,\beta,A,n}$ such that 
$\nu_{\alpha,\beta,A,n}(B)=\int_B k_{\alpha,\beta,A,n} d\mu$ for any 
$B \in {\cal A}$. We define the maps $k_{\alpha,\beta} : (A,x) \mapsto 
\sum_{n=1}^4 i^n k_{\alpha,\beta,A,n}(x)$. It follows from 
$\nu_{\alpha,\beta,A,n} \le 3\mu$ that $|k_{\alpha,\beta}(A,x)| \le 12$ for 
$\mu$-almost every $x$. The above results and eq.\eqref{Kteq} imply that 
eq.\eqref{fHs} is fulfilled by any simple probability measure. Consider any 
probability measure $w$ and the map $w'$ on ${\cal A}\times{\cal E}$ given 
by the right side of eq.\eqref{fHs} with the class $[\omega]$ corresponding to 
$w$. Corollary \ref{Corapp} implies that, for any $\epsilon>0$, there is a simple 
probability measure $\hat w$ such that $d(w,\hat w)<\epsilon$. For any event 
$A$ and effect $E$, one has 
$|{\cal O}_h(w)(A,E)-{\cal O}_h(\hat w)(A,E)|<\epsilon$, since the operation 
${\cal O}_h$ is continuous and obeys inequality \eqref{contineq}, see Lemma 
\ref{cont}, and $\Vert E \Vert_{op} \le 1$ \cite{BS}. This leads to 
$|{\cal O}_h(w)(A,E)-w'(A,E)|<\epsilon(1+12 \sum_{\alpha,\beta} 
\Vert L_{\beta}^\dag E L_{\alpha}\Vert_{op})$ and hence 
${\cal O}_h(w)=w'$. Let $\{ L'_{\alpha} \}_{\alpha}$ be any basis of the set 
of operators on ${\cal H}$. Using the expansion 
$L_{\alpha}=\sum_{\beta} \ell_{\alpha,\beta} L'_{\beta}$, equation 
\eqref{fHs} can be rewritten with $L_{\alpha}$ replaced by $L'_{\alpha}$ and 
$k_{\alpha,\beta}$ by $k'_{\alpha,\beta}=\sum_{\gamma,\delta}  
\ell_{\gamma,\alpha} \ell_{\delta,\beta}^* k_{\gamma,\delta}$.

Assume now that the transformation ${\cal O}_h$ is continuous and given by 
eq.\eqref{fHs} and consider $\hat \omega=\rho 1_B/\mu(B)$ where $\rho$ is 
any density operator and $B$ is any event with finite and nonvanishing 
$\mu(B)$. The probability measure $\hat w$ corresponding to $[\hat \omega]$ 
is transformed into ${\cal O}_h(\hat w) : (A,E) \mapsto \sum_{\alpha,\beta}
s_{\alpha,\beta}\mathrm{tr}
\big(L_{\alpha}\rho L_{\beta}^\dag E \big)/\mu(B)$ where 
$s_{\alpha,\beta}=\int_B k_{\alpha,\beta}(A,x) d\mu(x)$. As the above sum is 
real, it can be rewritten with $s_{\alpha,\beta}$ replaced by $S_{\alpha,\beta}
=(s^{\phantom{*}}_{\alpha,\beta}+s_{\beta,\alpha}^*)/2$, which are the 
elements of a self-adjoint matrix $S$. Thus, there are numbers 
$\lambda_{\alpha,\gamma}$ such that $S_{\alpha,\beta}
=\sum_{\gamma} \lambda_{\alpha,\gamma}\lambda^*_{\beta,\gamma}$ 
\cite{BS} and so ${\cal O}_h(\hat w)(A,E) = \sum_{\gamma}\mathrm{tr}
\big(K_{\gamma}(A,B)\rho K_{\gamma}(A,B)^\dag E \big)/\mu(B)$ for any 
event $A$ and effect $E$, where 
$K_{\gamma}(A,B)=\sum_{\alpha} \lambda_{\alpha,\gamma} L_{\alpha}$. 
It follows from eq.\eqref{fHs} that ${\cal O}_h$ is convex-linear. So, it 
transforms any simple probability measure according to eq.\eqref{Kteq}. Since 
it is moreover continous, the above theorem ensures that it is an operation.
\end{proof}
\subsection{Discrete classical subsystem and quantification of the correlations 
between the two subsystems}
As mentioned above, for a discrete classical subsystem $c_h$, it can always be 
assumed that the sample space $X$ is a subset of $\mathbb{N}$, $\cal A$ is 
the set of all subsets of $X$, and $\mu$ is the counting measure. In this case, 
a hybrid probability measure $w$ determines a unique state $\omega$, given 
by Corollary \ref{Cor1}, and the following proposition can be derived from the 
above theorem.
\begin{Prop}
Let $c_h$ be discrete. 

A hybrid operation ${\cal O}_h$ transforms any state $\omega$ into the state
\begin{equation}
x \mapsto \sum_{x' \in X} \sum_{\alpha} 
L_{\alpha}(x,x') \omega(x') L_{\alpha}(x,x')^\dag ,
\label{dss}
\end{equation}
where the sequence of maps $L_{\alpha}$ from $X^2$ to the set of bounded
operators on ${\cal H}$ satisfies, for any $x' \in X$, 
$\sum_{x \in X,\alpha} L_{\alpha}(x,x')^\dag L_{\alpha}(x,x')=I$.

For any such sequence, eq.\eqref{dss} defines a hybrid state transformation 
that corresponds to a hybrid operation.
\end{Prop}
A completely quantum description is possible here \cite{D}. Consider a Hilbert 
space ${\cal H}'$ with dimension equal to the cardinality of $X$ and denote as 
$\{ | x \rangle \}_{x \in X}$ an orthonormal basis of ${\cal H}'$. In the sums 
below, $x$, or $x'$, runs over $X$. The operator 
$\hat \omega = \sum_{x} \omega(x) \otimes |x \rangle \langle x |$ on 
${\cal H} \otimes {\cal H}'$ is positive with unit trace, i.e., it is a quantum 
state. Moreover, the quantum operation with Kraus operators 
$K_{\alpha,x,x'}= L_{\alpha}(x,x') \otimes |x \rangle \langle x' |$, with the 
indices $x$ and $x'$ running over $X$, changes $\hat \omega$ into the density 
operator $\sum_{x} \omega'(x) \otimes |x \rangle \langle x |$ where $\omega'$ 
is the hybrid state given by eq.\eqref{dss}.

Let us discuss the quantification of the correlations between the classical and 
quantum subsystems. A measure of these correlations cannot increase under 
any hybrid operation describing non-interacting subsystems \cite{PRA1}. Such 
an operation changes any hybrid state $\omega$ into 
$x \mapsto \sum_{x'} h(x,x') \sum_{\alpha} 
L_{\alpha} \omega(x') L_{\alpha}^\dag$ where the operators $L_{\alpha}$ 
and non-negative map $h$ fulfill 
$\sum_{\alpha} L_{\alpha}^\dag L_{\alpha}=I$ and $\sum_{x} h(x,x')=1$ 
for any $x' \in X$, see eq.\eqref{indop}. Within the quantum description 
mentioned above, the corresponding quantum operation is the composition of 
the local operations with Kraus operators $L_{\alpha} \otimes I'$, where $I'$ 
is the identity operator of ${\cal H}'$, and 
$I \otimes \sqrt{h(x,x')}  |x \rangle \langle x' |$, respectively. As a measure of 
correlations between $c_h$ and $q_h$, we propose the classical-quantum 
mutual information
$$I(\omega)= \sum_{x} \mathrm{tr} (\omega_x \log \omega_x) 
- \sum_{x} f_x \log f_x - \mathrm{tr} (\rho \log \rho) , $$  
where the argument $x$ is written as an index, 
$f : x \mapsto \mathrm{tr} \omega_x$ is the probability density function of 
$c_h$ and $\rho=\sum_{x} \omega_x$ is the state of $q_h$. It is equal to 
the Holevo quantity of the ensemble $\{ (f_x, \omega_x/f_x) \}_{x|f_x>0}$ 
\cite{Ho} and to the quantum mutual information of the effective quantum 
state $\hat \omega$ discussed above and so $I(\omega)
=S(\hat \omega \Vert \rho \otimes \sum_{x}  f_x |x \rangle \langle x |)$ 
where $S$ is the quantum relative entropy. As $S$ does not increase under 
any quantum operation applied to both its arguments \cite{L}, the measure 
$I$ meets the above requirement. It follows from Araki-Lieb inequality \cite{AL} 
that $I(\omega)$ cannot exceed twice the von Neumann entropy of $\rho$. 
Consequently, when $q_h$ is a bipartite system with finite-dimensional Hilbert
space, it can be shown that there is a tradeoff between the quantum
entanglement of its two parts and the correlations between $q_h$ and $c_h$ 
\cite{to1,to2,to3,to4}.
\begin{proof}
Consider any operation ${\cal O}_h$, any non-simple state $\omega$, and the 
sequence of simple states $\hat \omega_y$ given, for large enough $y$, by 
$\hat \omega_y(x)= \omega(x)/t_y$ if $x \le y$ and vanishes otherwise, where 
$t_y=\sum_{x \le y} \mathrm{tr} \omega(x)$. One has 
$\omega=t_y \hat \omega_y + (1-t_y) \tilde \omega_y$ where the map 
$\tilde \omega_y$, given by $\tilde \omega_y(x)= \omega(x)/(1-t_y)$ if 
$x > y$ and vanishes otherwise, is a state. Denote $\omega'$, 
$\hat \omega'_y$ and $\tilde \omega'_y$ the images by ${\cal O}_h$ of 
$\omega$, $\hat \omega_y$, and $\tilde \omega_y$, respectively. It follows 
from eq.\eqref{Kteq} that $\hat \omega'_y(x)=\sum_{x' \le y,\alpha}
L_{\alpha}(x,x') \omega(x') L_{\alpha}(x,x')^\dag/t_y$ where 
$L_{\alpha}(x,x')=K_{\alpha}(\{x\},\{x'\})$ and the sum converges in 
trace norm. The convex-linearity of ${\cal O}_h$ implies 
$\omega'-t_y \hat \omega'_y = (1-t_y) \tilde \omega'_y$. We denote a sum 
over the whole set $X$ by $\sum_{x}$. Since 
$\sum_x \Vert \tilde \omega'_y(x) \Vert=1$ and 
$\lim_{y \rightarrow \infty} t_y=1$, 
$\Vert \omega'(x)-t_y \hat \omega'_y (x) \Vert$ vanishes as 
$y \rightarrow \infty$ for any $x$. In other words, $\omega'(x)$ is given by 
eq.\eqref{dss} where the sums converge in trace norm. For a simple state 
$\omega$, the derivation is straightforward. Consider the state 
$\omega=\rho 1_{\{x' \}}$ with any density operator $\rho$ and $x' \in X$. 
Since the transformed state $\omega'$ fulfills 
$\sum_{x} \mathrm{tr} \omega'(x)=1$, one has $\sum_{x,\alpha} 
L_{\alpha}(x,x')^\dag L_{\alpha}(x,x')=I$, where the sum converges weakly.

Consider now any sequence of maps $L_{\alpha}$ from $X^2$ to the set of 
bounded operators on ${\cal H}$ obeying, for any $x' \in X$, 
$\sum_{x,\alpha} L_{\alpha}(x,x')^\dag L_{\alpha}(x,x')=I$. Any state 
$\omega$ can be written as $\omega=f\eta$ where $f$ is a nonnegative 
function on $X$ such that $\sum_x f(x)=1$ and $\eta(x)$ is a density operator 
for any $x$. Equation \eqref{dss} with $\omega$ replaced by 
$\eta(x)1_{\{ x \}}$ gives a state that we denote as $O_h(\eta(x)1_{\{ x \}})$. 
We define, for any $x$ and $y$ of $X$, $s_y=\sum_{x' \le y} f(x')$ and 
$\theta_{x,y}=\sum_{x' \le y} \sum_{\alpha} 
L_{\alpha}(x,x') \omega(x') L_{\alpha}(x,x')^\dag$ which belongs to 
${\cal T}_h$. This last sum can be rewritten as 
$\theta_{x,y}=\sum_{x' \le y} f(x') O_h(\eta(x')1_{\{ x' \}}) (x)$ and so 
$\Vert \theta_{x,z} - \theta_{x,y} \Vert \le |s_z-s_y|$. Thus, $(\theta_{x,y})_y$ 
is a Cauchy sequence, as $\lim_{y \rightarrow \infty} s_y=1$. It converges in 
${\cal T}_h$ which is a Banach space \cite{BS}. As its limit $\theta_{x,\infty}$ 
is positive and $\sum_x \mathrm{tr} \theta_{x,\infty}=1$, eq.\eqref{dss} 
defines a state transformation $O_h$. Since 
$\sum_x \Vert O_h(\omega_2)(x)-O_h(\omega_1)(x) \Vert \le 
\sum_x \Vert \omega_2(x)-\omega_1(x) \Vert $ for any states $\omega_1$ 
and $\omega_2$ and the simple probability measures are transformed 
according to eq.\eqref{Kteq} with the Kraus maps $K_{x,x',\alpha}$ given by 
$K_{x,x',\alpha}(A,B)=1_A(x) 1_B(x') L_{\alpha}(x,x')$ where $x$ and $x'$ 
run over $X$, $O_h$ corresponds to an operation.
\end{proof}
\subsection{Local operations and classical communication}
We are here concerned with the entanglement between two quantum systems, 
with Hilbert spaces ${\cal H}_1$ and ${\cal H}_2$. This entanglement vanishes 
if and only if their common state $\rho$ is the limit, in trace norm, of a 
sequence of states 
\begin{equation}
\rho_n = \sum_r p_{n,r} \rho_{1,n,r} \otimes  \rho_{2,n,r} , 
\label{ss}
\end{equation} 
where the sum is finite, the probabilities $p_{n,r}$ sum to unity, and 
$\rho_{s,n,r}$ denotes a density operator on ${\cal H}_s$ \cite{W}. Such a 
quantum state is said to be separable. A quantum operation that involves only 
local measurements, i.e., performed on one of the systems, and classical 
communication between the two systems, cannot transform a separable state 
into an entangled one. These transformations are named LOCC operations 
\cite{HHHH,LOCC}. 

They can be described within the present framework as follows. Consider a 
hybrid system with ${\cal H}={\cal H}_1\otimes {\cal H}_2$ and 
$X=\prod_{r=1}^n \{ 0,1, \ldots, y_r \}$, where $y_r$ is a positive integer, 
and denote $x=(x_1,\ldots,x_{n})$. The component $x_r$ is used to record 
the outcome of the $r$-th local measurement. The measurements outcomes 
are labeled by positive integers and $x_r$ is zero before the $r$-th 
measurement. A LOCC operation $\Lambda$ with $n$ rounds of 
communication can be achieved by applying to the hybrid state 
$x \mapsto \rho \prod_{r}\delta_{x_r,0}$ the following $n$ hybrid operations. 
The $r$-th operation transforms any hybrid state $\omega$ into the state 
$x \mapsto \sum_{x' \in X} L^{r}(x,x') \omega(x') L^{r}(x,x')^\dag$ where
$$L^{r}(x,x') =\prod_{s \neq r}\delta_{x_s,x'_s} 
V^{r}_{x_r} (x_1, \ldots, x_{r-1}) ,$$ with the last term replaced by 
$V^{1}_{x_1}$ when $r=1$. 

The $r$-th measurement is performed on the first (second) quantum subsystem 
when $r$ is odd (even) and so $V^{r}_{x_r}=V^{r,1}_{x_r} \otimes I_2$ 
($I_1 \otimes V^{r,2}_{x_r}$) where $I_s$ is the identity operator of 
${\cal H}_s$. The maps $V^{r}_{x_r}$ obey $V^{r}_0=0$ and 
$\sum_{x_r=0}^{y_r} (V^{r}_{x_r})^{\dag}V^{r}_{x_r}=I_1 \otimes I_2$. 
After the $r$-th operation, the hybrid state is 
$\omega_r : x \mapsto \prod_{s>r}\delta_{x_s,0} 
W_{x_1, \ldots, x_r}^{\phantom{\dag}} \rho W_{x_1, \ldots, x_r}^{\dag}$ 
where $$W_{x_1, \ldots, x_r} =\prod_{s=2}^r
V^{s}_{x_s} (x_1, \ldots, x_{s-1}) V^{1}_{x_1} . $$ In other words, $r$ 
quantum measurements have been performed, each depending on the 
outcomes of the preceding ones. The quantum state 
$\eta_r(x)=\omega_r(x)/\mathrm{tr}\omega_r(x)$ for 
$x=(x_1, \ldots, x_r,0,\ldots,0)$ is the density operator corresponding to the 
sequence of outcomes $x_1, \ldots, x_r$. The state of the quantum 
subsystems after the $r$ hybrid operations is $\Lambda(\rho)=\sum_{x \in X} 
W_{x}^{\phantom{\dag}} \rho W_{x}^{\dag}$ where $W_x$, given above, 
is a tensor product operator \cite{LOCC,PRA2}.

For finite-dimensional Hilbert spaces ${\cal H}_s$, the above discussed LOCC 
operations, involving discrete and finite classical information, allow to transform 
any quantum state, possibly entangled, into any separable one. From the point 
of view of quantum resource theories, this means that the minimal set of free 
states for these operations is the set of separable states \cite{QRT}. This can 
be seen as follows. First, observe that, in this case, any separable state can be 
written under the form of eq.\eqref{ss} \cite{HHHH}. Any quantum state can 
be transformed into 
$\Pi=|1\rangle_1 {_1\langle} 1| \otimes |1\rangle_2 {_2\langle} 1|$ using 
the quantum operation with Kraus operators 
$|1\rangle_1 {_1\langle} k| \otimes |1\rangle_2 {_2\langle} k'|$, where 
$\{ |k \rangle_s \}_k$ is an orthonormal basis of ${\cal H}_s$, which is the 
composition of two local operations. Then, this pure state can be turned into 
the separable state $\rho_n$ given by eq.\eqref{ss} using the LOCC operation 
with Kraus operators $\sqrt{p_{n,r}} L_{1,n,r,\alpha} \otimes L_{2,n,r,\beta}$ 
where $L_{s,n,r,\alpha}$ denotes those of a local operation changing 
$|1\rangle_s {_s\langle} 1|$ into $\rho_{s,n,r}$. For infinite-dimensional 
Hilbert spaces, on the other hand, there are separable states that cannot be 
cast into the form of eq.\eqref{ss}. LOCC operations cannot transform $\Pi$, 
for instance, into such states. 

The hybrid operations discussed above are special cases of those considered in 
Corollary \ref{CorO} that describe interacting classical and quantum systems in 
the presence of another quantum system which has no interaction with them 
and no intrinsic evolution. More precisely, under such transformations, 
any state 
$x \mapsto \omega_1(x) \otimes \rho_2$ ($\rho_1 \otimes \omega_2(x)$) 
where $\omega_s$ denotes a state of the hybrid system consisting of the 
classical subsystem and the first (second) quantum subsystem and $\rho_s$ 
denotes a density operator on ${\cal H}_s$, is transformed into 
$x \mapsto \omega'_1(x) \otimes \rho_2$ ($\rho_1 \otimes \omega'_2(x)$) 
where $\omega'_s$ is the image of $\omega_s$ by a hybrid operation. 
Operations of this kind, with any classical subsystem, discrete or not, should 
not generate quantum entanglement. 

We show below that applying any finite sequence of continuous such 
operations to an initial hybrid state $x \mapsto f(x) \rho$, where $f$ is any 
probability density function and $\rho$ is any density operator, gives a 
quantum subsystems state that is separable if $\rho$ is. Moreover, using such 
transformations, any initial state $\rho$ can be turned into any separable state 
given by
\begin{equation}
\rho_{sep} = \int f(x) \eta_1 (x) \otimes \eta_2(x) d\mu(x) , 
\label{rhosep}
\end{equation}
where $\eta_s$ is a map from $X$ to the set of density operators on 
${\cal H}_s$ such that $x \mapsto \mathrm{tr}(\eta_s(x) M)$ is 
$\mu$-integrable for any bounded operator $M$. The separability of 
$\rho_{sep}$ is shown below. The state $\rho$ can be transformed into 
$|1\rangle_1 {_1\langle} 1| \otimes |1\rangle_2 {_2\langle} 1|$ using the local 
operations seen above. Then, applying successively two continuous operations 
that turn any hybrid state $\omega$ into, respectively, 
$x \mapsto \sum_{\alpha} L_{1,\alpha}(x) \otimes I_2
\omega(x) L_{1,\alpha}(x)^\dag  \otimes I_2$ and 
$x \mapsto \sum_{\alpha} I_1 \otimes L_{2,\alpha}(x)
\omega(x) I_1 \otimes L_{2,\alpha}(x)^\dag$, where $L_{s,\alpha}(x)$ 
denotes Kraus operators of a quantum operation that changes 
$|1\rangle_s {_s\langle} 1|$ into $\eta_s(x)$, leads to the state given by 
eq.\eqref{rhosep}. These last two transformations can be understood as 
$x$-dependent operations performed on one of the quantum subsystems. 
An explicit expression for $L_{s,\alpha}(x)$ is given below. Since any 
separable state can be written under the form of eq.\eqref{rhosep} 
\cite{KSW}, any quantum state can be transformed into any separable one 
using the hybrid operations considered here.
\begin{proof}
Let $h$ be the hybrid system consisting of any classical system and the 
quantum system with Hilbert space ${\cal H}_1$ and ${\cal O}_h$ be any 
continuous operation of $h$. Let us name $q$ the quantum system with Hilbert 
space ${\cal H}_2$. Denote as ${\cal O}$ a continuous operation of $hq$ such 
that, for any $[\omega] \in \mathrm{P}_h$ and density operator $\rho$ of $q$, 
$O'([\omega \otimes \rho])=[\omega' \otimes \rho]$ where 
$\omega' \in O'_h([\omega])$ and $O'$ ($O'_h$) is the bounded linear map 
from $\mathrm{L}$ ($\mathrm{L}_h$) to itself corresponding to ${\cal O}$ 
(${\cal O}_h$). The existence of ${\cal O}$ is ensured by Corollary \ref{CorO}. 
For any $\epsilon>0$, we define the set $\mathrm{C}_{\epsilon}$ consisting 
of all classes $[\Omega]$ of $\mathrm{P}$ for which there is 
$\Omega_{\epsilon}=\sum_r c_r \rho_{1,r} \otimes \rho_{2,r} 1_{A_r}$, 
where the sum is finite, the positive coefficients $c_r$ and events $A_r$ obey 
$\sum_r p_r =1$ with $p_r=c_r \mu(A_r)$, and $\rho_{s,r}$ denotes a density 
operator on ${\cal H}_s$, such that 
$\vert [ \Omega - \Omega_{\epsilon} ]\vert < \epsilon$. The corresponding 
quantum subsystems state 
$\rho_{\epsilon}=\sum_r p_r \rho_{1,r} \otimes \rho_{2,r}$ satisfies 
$\Vert \rho - \rho_{\epsilon} \Vert < \epsilon$ where $\rho$ is that 
corresponding to $[\Omega]$ \cite{RN}.

Let us first show that 
$O'(\mathrm{C}_{\epsilon}) \subset \mathrm{C}_{2\epsilon}$. Consider any 
$[\Omega] \in \mathrm{C}_{\epsilon}$ and the corresponding map 
$\Omega_{\epsilon}$. It follows from Lemma \ref{cont} that 
$\vert [ \Omega' - \Omega'_{\epsilon} ]\vert < \epsilon$ where 
$\Omega' \in O'([\Omega])$ and 
$\Omega'_{\epsilon} \in O'([\Omega_{\epsilon}])$. The properties of 
${\cal O}$ imply 
$[\Omega'_{\epsilon}]=\sum_r p_r [\omega'_{1,r} \otimes \rho_{2,r}]$ 
with $\omega'_{1,r} \in O'_h([\omega_{1,r}])$, where 
$\omega_{1,r}=\rho_{1,r} 1_{A_r}/\mu(A_r)$. 
As $[\omega'_{1,r}]$ belongs to $\mathrm{P}_h$, there is a class 
$[\hat \omega_{1,r}]$ in $\mathrm{P}'_h$ such that 
$\vert [ \omega'_{1,r} - \hat \omega_{1,r} ]\vert < \epsilon$, see 
Corollary \ref{Corapp}. Denote 
$\Omega''_{\epsilon} = \sum_r p_r \hat \omega_{1,r} \otimes \rho_{2,r}$. 
It results from the above that 
$\vert [ \Omega' - \Omega''_{\epsilon} ]\vert < 2\epsilon$. One has 
$\hat \omega_{1,r} = \sum_s \omega_{1,r,s} 1_{A_{r,s}}$ with positive 
trace-class operators $\omega_{1,r,s}$ and events $A_{r,s}$ such that 
$\sum_s \mu(A_{r,s}) \mathrm{tr} \omega_{1,r,s}=1$, where the sums are 
finite, and so $\Omega''_{\epsilon}
=\sum_{r,s} c_{r,s}\rho_{1,r,s} \otimes \rho_{2,r} 1_{A_{r,s}}$ with 
$c_{r,s}=p_r \mathrm{tr} \omega_{1,r,s}$ and 
$\rho_{1,r,s}=\omega_{1,r,s}/\mathrm{tr} \omega_{1,r,s}$. Consequently, 
$O'([\Omega])$ belongs to $\mathrm{C}_{2\epsilon}$. The roles of the two 
quantum subsystems of $hq$ can be interchanged in the above. Consider any 
$\epsilon > 0$. By definition of separable states, the class of the hybrid state 
$x \mapsto f(x) \rho$, where $f$ is any probability density function and $\rho$ 
is any separable state, belongs to $\mathrm{C}_{\epsilon}$. Applying to it $n$ 
operations of the kind considered above gives a class in 
$\mathrm{C}_{\epsilon'}$ where $\epsilon'=2^n\epsilon$. The corresponding 
quantum subsystems state is hence separable.

Equation \eqref{rhosep} can be rewritten as 
$\rho_{sep} = \int \omega_1(x) \otimes \eta_2(x) d\mu(x)$ with 
$[\omega_1] \in \mathrm{P}_h$. For any $\epsilon>0$, there is 
$[\hat \omega_1] \in \mathrm{P}'_h$ such that 
$\vert [\omega_1-\hat \omega_1] \vert < \epsilon$, see Corollary \ref{Corapp}. 
Define $\hat \rho_{sep} = \int \hat \omega_1(x) \otimes \eta_2(x) d\mu(x)$. 
It follows from the properties of the Bochner integral and of the trace norm 
that $\Vert \rho_{sep} - \hat \rho_{sep} \Vert < \epsilon$ \cite{BS,RN}. One 
has $\hat \omega_{1} = \sum_r \omega_{1,r} 1_{A_{r}}$ with positive 
trace-class operators $\omega_{1,r}$ and events $A_{r}$ such that 
$\sum_r \mu(A_{r}) \mathrm{tr} \omega_{1,r}=1$, where the sums are finite. 
Thus, $\hat \rho_{sep}$ can be recast into the form 
$\hat \rho_{sep}=\sum_r p_r \rho_{1,r} \otimes \rho_{2,r}$ where 
$p_r=\mathrm{tr} \omega_{1,r}\mathrm{tr} \omega_{2,r}$ and 
$\rho_{s,r}=\omega_{s,r}/\mathrm{tr} \omega_{s,r}$ with 
$\omega_{2,r}=\int_{A_r} \eta_2 d\mu$. The probabilities $p_r$ sum to unity 
as $\mathrm{tr} \eta_2(x)=1$ for any $x$. Consequently, $\rho_{sep}$ is 
separable.

For any $x$, as $\eta_1(x)$ is a density operator, there exists an orthonormal 
basis $\{ | x,k \rangle \}_{k \ge 1}$ of ${\cal H}_1$ made of eigenvectors of 
$\eta_1(x)$ \cite{BS}. We denote the corresponding eigenvalues by 
$\lambda_{x,k}$ and define the Kraus operators 
$L_0(x)=I_1-|1 \rangle \langle 1 |$, where $|1 \rangle=|1 \rangle_1$, and 
$L_k(x)=\sqrt{\lambda_{x,k}} | x,k \rangle \langle 1 |$ for $k \ge 1$. They 
fulfill $\sum_k L_k(x)^{\dag}L_k(x)=I_1$. For any 
$[\Omega] \in \mathrm{P}$, the map 
\begin{multline}
\Omega' : x \mapsto 
\sum_k L_k(x) \otimes I_2 \Omega(x) L_k(x)^\dag \otimes I_2 \label{eq} \\
=\eta_1(x) \otimes \mathrm{tr}_1 \big(\Pi' \Omega(x) \big) 
+(I-\Pi') \Omega(x) (I-\Pi') ,  
\end{multline}
where $\mathrm{tr}_1$ denotes the partial trace with respect to ${\cal H}_1$, 
$\Pi'=|1 \rangle \langle 1 |\otimes I_2$, and $I=I_1 \otimes I_2$, is 
$\mu$-Bochner integrable and $[\Omega'] \in \mathrm{P}$. Thus, 
$O : [\Omega] \mapsto [\Omega']$ corresponds to a probability measure 
transformation ${\cal O}$ of $hq$. For 
$\Omega : x \mapsto |1 \rangle \langle 1 | \otimes \omega_2(x)$ with a map 
$\omega_2$ from $X$ to the set of positive trace-class operators on 
${\cal H}_2$, one has $\Omega' : x \mapsto \eta_1(x) \otimes \omega_2(x)$.
Consider any quantum system $q'$ with Hilbert space ${\cal H}_{q'}$. A 
probability measure transformation of $hqq'$ can be defined as above with the 
Kraus operators $L_k(x) \otimes I_2 \otimes I_{q'}$ where $I_{q'}$ is the 
identity operator of ${\cal H}_{q'}$. It fulfills eq.\eqref{vieq} with $\cal O$ 
since, for any effects $E$ and $F$ of $hq$ and $q'$, respectively, and density 
operator $\Omega$ on ${\cal H}_1 \otimes {\cal H}_2 \otimes {\cal H}_{q'}$, 
$\mathrm{tr}(\Omega E \otimes F)
=\mathrm{tr}(\mathrm{tr}_{q'}(\Omega I \otimes F) E)$ and 
$L_k(x)^\dag \otimes I_2 E L_k(x) \otimes I_2$ is an effect for any $x$. So, 
$\cal O$ is an operation. It is continuous as 
$\vert O ([\Omega_1])- O ([\Omega_2]) \vert 
\le 2 \vert [ \Omega_1-\Omega_2] \vert$ for any classes $[\Omega_1]$ 
and $[\Omega_2]$ of $\mathrm{P}$, see eq.\eqref{eq}.
\end{proof}
\section{Conclusion}\label{C}
We have proposed four requirements for a hybrid classical-quantum probability 
measure which are that it is non-negative, $\sigma$-additive with respect to 
either of its arguments, and appropriately normalized. It follows from these 
axioms that the state of a hybrid system consists of the probability measure of 
the classical subsystem and a map from its sample space to the set of states of 
the quantum subsystem. When the probability measure of the classical 
subsystem can be characterized by a probability density function with respect 
to a $\sigma$-finite reference measure, the hybrid states simplify into maps 
from the sample space to the set of positive trace-class operators of the 
quantum subsystem. This last kind of states is commonly used in hybrid 
approaches \cite{O,KC,T1,ABCMP,A,D, BBGBT}. In this case, we have 
introduced a metric for the hybrid probability measures and shown that any 
of them can be approximated with any accuracy, with respect to this metric, 
by one corresponding to a simple hybrid state.

We have then considered particular transformations of hybrid probability 
measures that we name hybrid operations. The assumption defining them 
can be formulated as follows. For any hybrid operation and any ancillary 
quantum system, with finite-dimensional Hilbert space, there is a probability 
measure transformation of the bipartite system made up of the considered 
hybrid system and the ancillary that describes the action of the operation on 
the hybrid in the presence of the ancillary with no interaction between the two 
systems and no evolution of the ancillary. We have derived an expression for 
any such hybrid operation on the subset of probability measures corresponding 
to simple hybrid states. Continuous operations, with respect to the considered 
metric, are fully determined by this expression. As shown, all hybrid operations 
are continuous when the classical subsystem is discrete or the Hilbert space of 
the quantum subsystem is finite-dimensional. 

From the general expression mentioned above, we have obtained explicit 
expression for hybrid operations, on the whole set of probability measures 
of the considered hybrid system, when the classical and quantum subsystems 
are non-interacting, the Hilbert space of the quantum subsystem is 
finite-dimensional, or the classical subsystem is discrete. In some particular 
cases, the corresponding hybrid state transformations are identical to previously 
considered ones \cite{O,OSSW,D}. We have also discussed the quantification of 
the correlations between the classical and quantum subsystems and a 
generalization of LOCC operations that takes into account non-discrete classical 
information.

Our results constrain the properties of hybrid descriptions. To achieve a fully 
consistent hybrid dynamics, further steps are needed. 
For the quantum degrees of freedom, those of a quantum field for instance, 
our approach assumes a separable Hilbert space. Therefore, it holds only if the 
quantum subsystem can be described in such terms. The hybrid states and 
operations we have obtained follow from a small number of axioms for hybrid 
probability measures and from the hybrid operation definition. Additional 
physical requirements, not considered here, may force to reject some hybrid 
operations. Moreover, in order to formulate the dynamics of a hybrid system, 
it is necessary to consider a time-parameterized family of operations that 
fulfills conditions similar to those leading to the Lindblad quantum master 
equation \cite{L2}.

\end{document}